\documentclass[conference]{IEEEtran}
\usepackage{graphicx}
\usepackage{amsthm}
\usepackage{epsfig}
\usepackage{latexsym}
\usepackage{amsfonts}
\usepackage{here}
\usepackage{rawfonts}
\usepackage[latin1]{inputenc}
\usepackage[T1]{fontenc}
\usepackage{calc}
\usepackage{capitalgreekitalic}
\usepackage{url}
\usepackage{enumerate}
\usepackage{color}
\usepackage[tbtags]{amsmath}
\usepackage{amssymb}
\usepackage{upref}
\usepackage{epic,eepic}
\usepackage{times}
\usepackage{dsfont}
\usepackage{comment}
\usepackage{cite}
\usepackage{bbm} 

\renewcommand{\Pr}{\ensuremath{\operatorname{Pr}}}

\newtheorem{corollary}{Corollary}
\newtheorem{theorem}{\bf Theorem}
\newtheorem{proposition}{\bf Proposition}

\newtheorem{definition}{\bf Definition}

\usepackage{dsfont}

\newcounter{step}
\newlength{\totlinewidth}
  {\end{list}%
  \rule{\linewidth}{1pt}}
\newcounter{substep}

  {\end{list}}

\newlength{\aligntop}
\newlength{\alignbot}
 \makeatletter
\renewenvironment{align}{%
  \vspace{\aligntop}
  \start@align\@ne\st@rredfalse\m@ne
}{%
  \math@cr \black@\totwidth@
  \egroup
  \ifingather@
    \restorealignstate@
    \egroup
    \nonumber
    \ifnum0=`{\fi\iffalse}\fi
  \else
    $$%
  \fi
  \ignorespacesafterend%
  \vspace{\alignbot}\par\noindent
} \makeatother

\IEEEoverridecommandlockouts

\usepackage{algorithm}
\usepackage{algorithmic}

\usepackage{subfigure}

\begin{document}
\title{\huge Virtual Reality over Wireless Networks: Quality-of-Service Model and Learning-Based Resource Management}
\author{{Mingzhe Chen\IEEEauthorrefmark{1}}, Walid Saad\IEEEauthorrefmark{2}, and Changchuan Yin\IEEEauthorrefmark{1}\vspace*{0em}\\ 
\authorblockA{\small \IEEEauthorrefmark{1}Beijing Laboratory of Advanced Information Network, Beijing University of Posts and Telecommunications, Beijing, China 100876,\\ Emails: \protect\url{chenmingzhe@bupt.edu.cn}, \protect\url{ccyin@ieee.org.} \\
\IEEEauthorrefmark{2}Wireless@VT, Bradley Department of Electrical and Computer Engineering, Virginia Tech, Blacksburg, VA, USA, Email: \protect\url{walids@vt.edu.}\\
}\vspace*{-3em}
\thanks{{
This work was supported in part by the National Natural Science Foundation of China under Grant 61671086 and Grant 61629101, the 111 Project under Grant B17007, the Director Funds of Beijing Key Laboratory of Network System Architecture and Convergence under Grant 2017BKL-NSAC-ZJ-04, the Beijing Natural Science Foundation under Grant L172032, the BUPT Excellent Ph.D. Students Foundation, and in part by the U.S. National Science Foundation under Grant CNS-1460316.}}
 }

\maketitle

\vspace{0cm}
\begin{abstract}

In this paper, the problem of resource management is studied for a network of wireless virtual reality (VR) users communicating over {small cell networks} (SCNs). In order to capture the VR users' quality-of-service (QoS) in SCNs, a novel VR model, based on multi-attribute utility theory, is proposed. This model jointly accounts for VR metrics such as tracking accuracy, processing delay, and transmission delay. In this model, the small base stations (SBSs) act as the VR control centers that collect the tracking information from VR users over the cellular uplink. Once this information is collected, the SBSs will then send the three-dimensional images and accompanying audio to the VR users over the downlink. Therefore, the resource allocation problem in VR wireless networks must jointly consider both the uplink and downlink. This problem is then formulated as a noncooperative game and a distributed algorithm based on the machine learning framework of echo state networks (ESNs) is proposed to find the solution of this game. The proposed ESN algorithm enables the SBSs to predict the VR QoS of each SBS and is guaranteed to converge to a mixed-strategy Nash equilibrium. The analytical result shows that each user's VR QoS jointly depends on both VR tracking accuracy and wireless resource allocation. Simulation results show that the proposed algorithm yields significant gains, in terms of VR QoS utility, that reach up to {22.2\%} and {37.5\%}, respectively, compared to Q-learning and a baseline proportional fair algorithm. The results also show that the proposed algorithm has a faster convergence time than Q-learning and can guarantee low delays for VR services.
\end{abstract}

{\renewcommand{\thefootnote}{\fnsymbol{footnote}}
\footnotetext{A preliminary version of this work was published in the IEEE GLOBECOM conference \cite{QoSmodelR}.}}

\section{Introduction}
Virtual reality (VR) services will enable users to experience and interact with virtual and immersive environments through a first-person view \cite{rosedale2017virtual}. 
For instance, individuals can use a VR device to walk around in a fully immersive world and travel to any destination, within the confines of their own home.
{Compared to a static high definition (HD) video, a VR video is generated based on the users' movement such as head and eye movements. Therefore, generating a VR video requires tracking information related to the users' interactions with the VR environment. In consequence, the tracking accuracy and the delay of tracking information transmission will significantly affect the creation and transmission of VR videos hence affecting the users' experience. Such tracking delay does not exist for regular HD videos and, thus, constitutes a fundamental difference between VR and HD videos.}
If VR devices such as HTC Vive \cite{htc} rely on wired connections to a VR control center, such as a computer, for processing their information, then the users will be significantly restricted in the type of actions that they can take and the VR applications that they can experience.
To enable truly immersive VR applications, one must deploy VR systems \cite{bacstuug2016towards} over wireless cellular networks. 
In particular, VR systems can use the wireless connectivity of small cell networks (SCNs) \cite{bacstuug2016towards} in which small cell base stations (SBSs) can act as VR control centers that directly connect to the VR devices over wireless links, collect the tracking\footnote{{Here, tracking pertains to the fact that the immersive VR applications must continuously collect a very accurate localization of each user including the positions, orientation, and eye movement (i.e., gaze tracking).}} information from the VR devices, and send the VR images to the VR devices over wireless links. However, operating VR devices over SCNs faces many challenges in terms of tracking, low delay {(typically less than 20 ms)}, and high data rate {(typically over 25 Mbps)} \cite{bacstuug2016towards}.
\subsection{Related Work}

The existing literature has studied a number of problems related to VR such as in \cite{bacstuug2016towards,rosedale2017virtual,6849494, 7358172,qian2016optimizing,richardt2013megastereo,sridhar2016real,rovira2017reinforcement,Pfeiffer,lo2017performance,abari2016cutting}. The authors in  \cite{rosedale2017virtual} and \cite{bacstuug2016towards} provided qualitative surveys that motivate the deployment of VR over wireless systems, but these works do not provide any mathematically rigorous modeling. 
{In \cite{6849494}. the authors proposed a distortion-aware concurrent multipath data transfer algorithm to minimize the end-to-end video distortion. The work in \cite{7358172} deveploped a modeling-based approach to optimize the high frame rate video transmission over wireless networks. However, the works in \cite{6849494} and \cite{7358172} only consider the video content transmission which cannot be directly applied to the VR content transmission since the VR contents are generated based on the users' tracking information.}
In \cite{qian2016optimizing}, the authors proposed a streaming scheme that delivers only the visible portion of a $360^\circ$ video based on head movement prediction. Meanwhile, the authors in \cite{richardt2013megastereo} developed an algorithm for generating high-quality stereo panoramas. The work in \cite{sridhar2016real} proposed a real-time solution that uses a single commodity RGB-D camera to track hand manipulation. In \cite{rovira2017reinforcement}, a reinforcement learning algorithm is proposed to guide a user's movement within a VR immersive environment. The authors in \cite{Pfeiffer} proposed an approach based on the three-dimensional
 (3D) heat maps to address the delay challenges. The work in \cite{lo2017performance} designed several experiments for quantifying the performance of tile-based $360^\circ$ video streaming over a real cellular network. In \cite{abari2016cutting}, the authors performed a WiFi experiment for wireless VR for a single user within a single room. 
However, beyond the survey in \cite{bacstuug2016towards}, which motivated the use of VR over wireless and the WiFi experiment for a single VR user in \cite{abari2016cutting}, the majority of existing VR works in \cite{rosedale2017virtual,qian2016optimizing,richardt2013megastereo,sridhar2016real,rovira2017reinforcement,Pfeiffer,lo2017performance} focus on VR systems that are deployed over wired networks and, as such, they do not capture any challenges of deploying VR over wireless SCNs. Moreover, most of these existing works \cite{bacstuug2016towards,rosedale2017virtual,qian2016optimizing,richardt2013megastereo,sridhar2016real,rovira2017reinforcement,Pfeiffer,lo2017performance,abari2016cutting} focus only on improving a single VR quality-of-service (QoS) metric such as tracking or generation of 3D images. Indeed, this prior art does not develop any VR-specific model that can capture all factors of VR QoS and, hence, these existing works fall short in addressing the challenges of optimizing VR QoS for wireless users.

\subsection{Main Contributions} 
The main contribution of this paper is a novel framework for enabling wireless cellular networks to integrate VR applications and services. To our best knowledge, \emph{this is the first work that develops a comprehensive framework for analyzing the performance of VR services over cellular networks.} Our main contribution include:
\begin{itemize}
\item We propose a novel VR model based on \emph{multi-attribute utility theory} \cite{abbas2010constructing}, to jointly capture the tracking accuracy, transmission delay, and processing delay thus effectively quantifying the VR QoS for all wireless users. In this VR model, the tracking information is transmitted from the VR users to the SBSs over the cellular uplink while the VR images are transmitted in the downlink from the SBSs to their users. 

\item We analyze resource (resource blocks) allocation \emph{jointly} over the uplink and downlink. We formulate the problem as a noncooperative game in which the players are the SBSs. Each player seeks to find an optimal spectrum allocation scheme to optimize a utility function that captures the VR QoS. 
\item To solve this VR resource management game, we propose a learning algorithm based on echo state networks (ESNs) \cite{chen2016caching} that can predict the VR QoS value resulting from resource allocation and reach a mixed-strategy Nash equilibrium (NE). { Compared to existing learning algorithms \cite{chekroun2015distributed,bennis2011distributed,zhu2011eavesdropping}, the proposed algorithm has lower complexity, requires less information due to its neural network nature, and is guaranteed to converge to a mixed-strategy NE.}  
One unique feature of this algorithm is that, after training, it can use the stored ESN information to effectively find an optimal converging path to a mixed-strategy NE.
\item We perform fundamental analysis on the gains and tradeoffs that stem from changing the number of uplink and downlink resource blocks for each user. This analytical result shows that, in order to improve the VR QoS of each user, we can improve the tracking system or increase the number of the resource blocks allocated to each user according to each user's specific state. 
\item Simulation results show that the proposed algorithm can yield, respectively, {22.2\%} and {37.5\%} gains in terms of VR QoS utility compared to Q-learning and proportional fair algorithm. The results also show that the proposed algorithm significantly improves the convergence time of up to {24.2\%} compared to Q-learning.  
\end{itemize}


 
 The rest of this paper is organized as follows. The problem is formulated in Section \uppercase\expandafter{\romannumeral2}. The resource allocation algorithm is proposed in Section \uppercase\expandafter{\romannumeral3}. In Section \uppercase\expandafter{\romannumeral4}, simulation results are analyzed. Finally, conclusions are drawn in Section \uppercase\expandafter{\romannumeral5}.               

\section{System Model and Problem Formulation}\label{se:system}

%

Consider the downlink transmission of a cloud-based small cell network (SCN)\footnote{{Since next-generation cellular networks will all rely on a small cell architecture \cite{6824752}, we consider SCNs as the basis for our analysis. However, the proposed VR model and algorithm can also be used for any other type of cellular networks.}}
 servicing a set $\mathcal{U}$ of $V$ wireless VR users via a set $\mathcal{B}$ of $B$ SBSs \cite{7412759}. Here, we focus on entertainment VR application such as watching immersive videos and playing immersive games \cite{qian2016optimizing}. 
VR allows users 
to be immersed in a virtual environment within which they can experience a 3D and high-resolution $360^\circ$ vision with 3D surround stereo. In particular, immersive VR will provide a $360^\circ$ panoramic image for each eye of a VR user. Compared to a conventional $120^\circ$ image, a $360^\circ$ panoramic image enables a VR user to have a surrounded vision without any dead spots. However, a $360^\circ$ VR image needs more pixels than a traditional two-dimensional image, and, hence, VR transmission will require more stringent data rate \cite{bacstuug2016towards} and delay (less than 20 ms) requirements than traditional multimedia services. Moreover, for an HD video, we only need to consider the  video transmission delay as part of the QoS. In contrast, for a VR video, we need to jointly consider the video transmission delay, tracking information transmission delay, and the delay of generating the VR videos based on the tracking information. {Note that, for VR applications, the transmission delay and processing delay will directly determine how each VR video is transmitted and, thus, they are a key determinant of the VR video quality.}


\begin{table*}\footnotesize
  \newcommand{\tabincell}[2]{\begin{tabular}{@{}#1@{}}#2\end{tabular}}
\renewcommand\arraystretch{0.9}
 \caption{
    \vspace*{-0.3em}List of notations}\vspace*{-1em}
\centering  
\begin{tabular}{|c||c|c||c|}
\hline
Notation & Description & Notation & Description\\
\hline
$V$ & Number of users & $S^\textrm{d}, S^\textrm{u}$ & Number of downlink and uplink resource blocks\\
\hline
$B$ & Number of SBSs & $\mathcal{S}^\textrm{d}, \mathcal{S}^\textrm{u}$ & Sets of downlink and uplink resource blocks \\
\hline
$\boldsymbol{W}_\textrm{in}$& Input weight matrix  & $D_{ij}^\textrm{T}$ & Transmission delay between user $i$ and SBS $j$ \\
\hline
 $P_B$ & Transmit power of SBSs& $d_{ij}$& Distance between user $i$ and SBS $j$\\
\hline
${\mathcal{A}_j}$ & Set of SBS $j$'s actions &  $c_{ij}$ & Data rate of user $i$ associated with SBS $j$ \\
\hline
$\boldsymbol{a}_j$ & One action of SBS $j$ & $\boldsymbol{s}_{ij}^\textrm{d},\boldsymbol{s}_{ij}^\textrm{u} $ & Resource allocation vector\\
\hline
$\boldsymbol{a}_{ji}$ & Action $i$ of SBS $j$ & $\boldsymbol{d}_{j}$ & Downlink resource allocation of SBS $j$ \\
\hline
$\boldsymbol{W}_\textrm{out}$& Output weight matrix & $\bar u_j$ & Average utility function of SBS $j$ \\
\hline
$N_w$ & Number of reservoir units & $\hat u_j$ & Utility function of SBS $j$\\
\hline
$K_i$ & Tracking accuracy& $\gamma_{K}$ & Maximal tracking inaccuracy\\
\hline
$D_i$& Total delay of user $i$&$\boldsymbol{v}_{j}$ & Uplink resource allocation of SBS $j$\\
\hline
  $\boldsymbol{\chi}$& Vector of user's localization & \footnotesize{$U_i\left( {{D_i},{K_i}} \right)$} & Total utility function of user $i$'s VR QoS \\
\hline
$\lambda $ & Learning rate & $U_i\left( {{D_i}\left| {{K_i}} \right.} \right)$ & Conditional utility function of user $i$'s VR QoS\\
\hline
 $D_i^\textrm{P}$ & Processing delay of user $i$ & $V_j$& Number of users associated with SBS $j$ \\
\hline
$\boldsymbol{W}$ & Reservoir weight matrix & $\left|{\mathcal{A}_j}\right|$ & Number of SBS $j$'s actions\\
\hline
$\boldsymbol{\mu}_j$& Reservoir state of SBS $j$ & $\boldsymbol{a}_{-j}$ & Actions of all SBSs other than SBS $j$ \\
\hline
$A$& Data size of tracking vector & $L$ & Maximum Data size of each VR image \\
\hline
\end{tabular}
\vspace{-0.4cm}
\end{table*}  

\begin{figure}[!t]
  \begin{center}
   \vspace{0cm}
    \includegraphics[width=9cm]{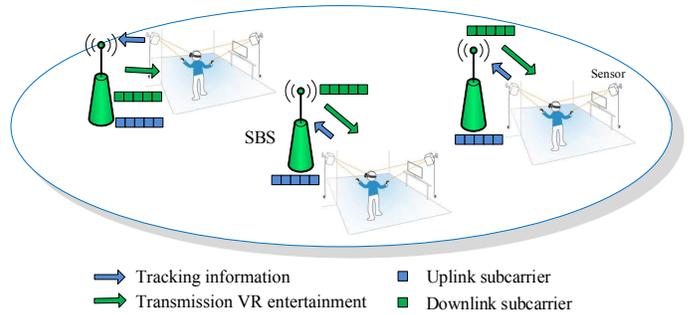}
    \vspace{-0.6cm}
    \caption{\label{figure2} A network of immersive VR application.}
  \end{center}\vspace{-0.7cm}
\end{figure}

The SBSs adopt an orthogonal frequency division multiple access (OFDMA) technique and transmit over a set of uplink resource blocks  $\mathcal{S}^\textrm{u}$ and a set of downlink resource blocks $\mathcal{S}^\textrm{d}$, as shown in Fig. \ref{figure2}. The uplink resource blocks are used to transmit the data that is collected by the VR sensors placed at a VR user's headset or near the VR user while the downlink resource blocks are used to transmit the image displayed on each user's VR device. We define the coverage of each SBS as a circular area of radius $r_B$ and we assume that each SBS will only allocate resource blocks to the users located in its coverage range. Table \uppercase\expandafter{\romannumeral1} provides a summary of the notations used hereinafter.

\subsection{VR Model}
For our VR model, we consider delay and tracking accuracy as the main VR QoS metrics of interest. Based on the accurate localization of each user, the SBS can build the immersive and virtual environment for each user. Among the components of the VR QoS, a delay metric can be defined to capture two key VR service requirements: high data rate and low delay. Next, we will explicitly discuss all the components of the considered VR QoS metrics. 
\subsubsection{Tracking Model}          
For any VR QoS metric, tracking consists of position tracking and orientation tracking \cite{qian2016optimizing}. VR tracking directly affects the construction of the users' virtual environment since the SBSs need to use each user's localization information to construct the virtual environment. Hereinafter, we use the term ``localization {information}'' to represent the information related to the user's location and orientation. The localization data of each user is used as the primary component of tracking \cite{bacstuug2016towards}.  The tracking vector of each user $i$ is $\boldsymbol{\chi}_i={\left[\chi_{i1},\chi_{i2},\chi_{i3}, \chi_{i4}, \chi_{i5}, \chi_{i6}\right]}$, where $\left[\chi_{i1},\chi_{i2},\chi_{i3}\right]$ represents the position of each VR user while $\left[ \chi_{i4}, \chi_{i5}, \chi_{i6} \right]$ represents the orientation of each user. Here, we note that the position and orientation of each user are determined by the SBS via the information collected by the sensors. For each VR user $i$, the tracking accuracy $K_i\left(\boldsymbol{s}_{ij}^\textrm{u}\right)$ can be given by:
{\begin{equation}\label{eq:Ri}
K_i\left(\boldsymbol{s}_{ij}^\textrm{u}\right)=
1-\frac{{\left\| {{\boldsymbol{\chi} _i}\left(\boldsymbol{ \gamma}_i^\textrm{u}\left({\boldsymbol{s}_{ij}^\textrm{u}} \right)\right) - {\boldsymbol{\chi} _i^R}} \right\|}}{{\mathop {\max }\limits_{\boldsymbol{s}_{ij}^\textrm{u}} \left\| {{\boldsymbol{\chi} _i}\left( \boldsymbol{ \gamma}_i^\textrm{u}\left({\boldsymbol{s}_{ij}^\textrm{u}}\right) \right) - {\boldsymbol{\chi} _i^R}} \right\|}},
\end{equation}
where ${\boldsymbol{\chi} _i}\left( \boldsymbol{ \gamma}_i^\textrm{u}\left({\boldsymbol{s}_{ij}^\textrm{u}} \right)\right)$ is the tracking vector transmitted from the VR headset to the SBS over wireless links that depends on the signal-to-interference-plus-noise (SINR) ratio. ${\boldsymbol{s}_{ij}^\textrm{u}}= \left[ {s_{ij1}^\textrm{u}, \ldots ,s_{ij{S^\textrm{u}}}^\textrm{u}} \right]$ represents the vector of uplink resource blocks that SBS $j$ allocates to user $i$ with $s_{ijk}^\textrm{u} \in \left\{ {1,0} \right\}$. Here, the uplink (downlink) resource blocks are equally divided into ${S^\textrm{u}}$ (${S^\textrm{d}}$) groups. Hereinafter, the term \emph{resource block} refers to one of those groups. $s_{ijk}^\textrm{u}=1$ indicates that resource block $k$ is allocated to user $i$. $\boldsymbol{\gamma}_i^\textrm{u}\left({\boldsymbol{s}_{ij}^\textrm{u}} \right)= \left[ {\gamma\left(s_{ij1}^\textrm{u}\right), \ldots ,\gamma\left(s_{ij{S^\textrm{u}}}^\textrm{u}\right)} \right]$ represents the SINR of SBS $j$ with resource blocks ${\boldsymbol{s}_{ij}^\textrm{u}}$ {where  ${\gamma \left({s}_{ijk}^\textrm{u} \right)=\frac{{{P_U}{h_{ij}^k}}}{{{\sigma ^2} + \sum\limits_{l \in \mathcal{U}_k,l \ne i} {{P_U}{h_{il}^k}} }}}$ is the SINR between user $i$ and SBS $j$ over resource block $k$. Here, $\mathcal{U}_k$ is the set of users that use uplink resource block $k$. $P_U$ is the transmit power of user $i$ (assumed to be equal for all users), $\sigma^2$ is the variance of the Gaussian noise and $h_{ij,\tau}^k=g_{ij,\tau}^kd_{ij}^{-\beta}$ is the path loss between user $i$ and SBS $j$ over resource block $k$ with $g_{ij,\tau}^k$ being the Rayleigh fading parameter, $d_{ij}$ the distance between user $i$ and SBS $j$, and $\beta$ the path loss exponent.}  In ${\boldsymbol{\chi} _{i}}\left( \boldsymbol{ \gamma}_i^\textrm{u}\left({\boldsymbol{s}_{ij,\tau}^\textrm{u}} \right)\right)$, the resource block vector $\boldsymbol{s}_{ij}^\textrm{u}$ determines the uplink SINR of the tracking vector's transmission. 
The tracking vector will be subject to bit errors and, hence, it will depend on the SINR and the resource block vector $\boldsymbol{s}_{ij}^\textrm{u}$. 
${\left\| {{\boldsymbol{\chi} _{i}}\left(\boldsymbol{ \gamma}_i^\textrm{u}\left({\boldsymbol{s}_{ij}^\textrm{u}} \right)\right) - {\boldsymbol{\chi} _{i}^R}} \right\|}$ represents the wireless tracking inaccuracy. ${\boldsymbol{\chi} _{i}^R}$ is obtained from the users' force feedback and transmitted via a dedicated wireless channel. Force feedback \cite{burdea2003virtual} represents the feedback that the users send to the SBSs whenever they are not satisfied with the displayed VR image. 
 {${\boldsymbol{\chi} _{i}^R}$ is only transmitted from the users to the SBSs when the users feel uncomfortable in the environment. Since ${\boldsymbol{\chi} _{i}^R}$ is not transmitted every time slot, the SBSs can use orthogonal resource blocks over a dedicated channel to transmit ${\boldsymbol{\chi} _{i}^R}$.} Thus, ${\boldsymbol{\chi} _{i}^R}$ is generally much more accurate than ${\boldsymbol{\chi} _{i}}\left( \boldsymbol{ \gamma}_i^\textrm{u}\left({\boldsymbol{s}_{ij}^\textrm{u}} \right)\right)$.
Note that, (\ref{eq:Ri}) is based on normalized root mean square errors \cite{APractical}, which is a popular measure of the difference between two datasets. From (\ref{eq:Ri}), we can see that, {when the SINR increases, the bit error rate will decrease and, hence, ${\left\| {{\boldsymbol{\chi} _{i}}\left(\boldsymbol{ \gamma}_i^\textrm{u}\left({\boldsymbol{s}_{ij}^\textrm{u}} \right)\right) - {\boldsymbol{\chi} _{i}^R}} \right\|}$ decreases and the tracking accuracy improves.}

\subsubsection{Delay} Next, we define the delay component that consists of the transmission delay and processing (and computing) delay. The transmission delay of each user $i$ will be:
\begin{equation}\label{eq:dt}
{
D_{ij}^\textrm{T}\left(\boldsymbol{s}_{ij}^\textrm{d},\boldsymbol{s}_{ij}^\textrm{u}\right) = \frac{{{L}}}{{{c_{ij}\left(\boldsymbol{s}_{ij}^\textrm{d}\right)}}}+ \frac{{{A}}}{{{c_{ij}\left(\boldsymbol{s}_{ij}^\textrm{u}\right)}}},}
\end{equation}
where $L$ is the {maximum} size of VR image that each SBS needs to transmit to its users, {$A$ is the size of the tracking vector that each user needs to transmit to the associated SBS, and ${c_{ij}\left(\boldsymbol{s}_{ij}^\textrm{d}\right)} = \sum\limits_{k = 1}^{{S^\textrm{d}}} {s_{ijk}^\textrm{d}B_\textrm{R}{{\log }_2}\left( {1 + \gamma\left(s_{ijk}^\textrm{d}\right)} \right)}$ is the downlink rate of user $i$. ${c_{ij}\left(\boldsymbol{s}_{ij}^\textrm{u}\right)\!=\!  \sum\limits_{k = 1}^{{S^\textrm{u}}} {s_{ijk}^\textrm{u}B_\textrm{R}{{\log }_2}\left( {1\! +\! \gamma\left(s_{ijk}^\textrm{u} \right)} \right)}}$ is the uplink rate of user $i$.} $\boldsymbol{s}_{ij}^\textrm{d}\! =\! \left[ {s_{ij1}^\textrm{d}, \ldots ,s_{ij{S^\textrm{d}}}^\textrm{d}} \right]$ is the vector of resource blocks that SBS $j$ allocates to user $i$ with $s_{ijk}^\textrm{d} \in \left\{ {0,1} \right\}$. $s_{ijk}^\textrm{d}=1$ indicates that resource block $k$ is allocated to user $i$. ${\gamma \left(s_{ijk}^\textrm{d} \right)=\frac{{{P_B}{h_{ij}^k}}}{{{\sigma ^2} + \sum\limits_{l \in \mathcal{B}_k,l \ne j} {{P_B}{h_{il}^k}} }}}$ is the signal-to-interference-plus-noise ratio between user $i$ and SBS $j$ over resource block $k$ with $\mathcal{B}_k$ being the set of the SBSs that use downlink resource block $k$. $B_\textrm{R}$ is the bandwidth of each resource block, $P_B$ is the transmit power of SBS $j$ (assumed to be equal for all SBSs). 

In the VR QoS, the \emph{processing delay} primarily stems from the tracking accuracy.
{To properly capture the processing delay, we define a vector $\boldsymbol{l}\left({\boldsymbol{\chi} _i}\left( \boldsymbol{ \gamma}_i^\textrm{u}\left({\boldsymbol{s}_{ij}^\textrm{u}} \right)\right)\right)$ that represents a VR image of user $i$ constructed by the tracking vector ${\boldsymbol{\chi} _i}\left( \boldsymbol{ \gamma}_i^\textrm{u}\left({\boldsymbol{s}_{ij}^\textrm{u}} \right)\right)$. 
{$\boldsymbol{l}\left({\boldsymbol{\chi} _i}\left( \boldsymbol{ \gamma}_i^\textrm{u}\left({\boldsymbol{s}_{ij}^\textrm{u}} \right)\right)\right)$ will be typically generated before the SBSs receive the force feedback tracking vector ${\boldsymbol{\chi} _i^R}$, as the VR system will use the historical tracking information to predict the future tracking vector. When the SBSs receive ${\boldsymbol{\chi} _i^R}$, they must construct $\boldsymbol{l}\left({\boldsymbol{\chi} _i^R}\right)$ based on ${\boldsymbol{\chi} _i^R}$. The simplest way of constructing the VR image $\boldsymbol{l}\left({\boldsymbol{\chi} _i^R}\right)$ is to correct it from $\boldsymbol{l}\left({\boldsymbol{\chi} _i}\left( \boldsymbol{ \gamma}_i^\textrm{u}\left({\boldsymbol{s}_{ij}^\textrm{u}} \right)\right)\right)$ to $\boldsymbol{l}\left({\boldsymbol{\chi} _i^R}\right)$.  }   
The processing delay will then represent the time that an SBS must spend to change the VR image from $\boldsymbol{l}\left({\boldsymbol{\chi} _i}\left( \boldsymbol{ \gamma}_i^\textrm{u}\left({\boldsymbol{s}_{ij}^\textrm{u}} \right)\right)\right)$ to $\boldsymbol{l}\left({\boldsymbol{\chi} _i^R}\right)$.   
The bits that the SBSs use to change the VR image from $\boldsymbol{l}\left({\boldsymbol{\chi} _i}\left( \boldsymbol{ \gamma}_i^\textrm{u}\left({\boldsymbol{s}_{ij}^\textrm{u}} \right)\right)\right)$ to $\boldsymbol{l}\left({\boldsymbol{\chi} _i^R}\right)$ can be calculated by using the motion search algorithm \cite{yang2002computation}. The motion search algorithm can find the different pixels between $\boldsymbol{l}\left({\boldsymbol{\chi} _i}\left( \boldsymbol{ \gamma}_i^\textrm{u}\left({\boldsymbol{s}_{ij}^\textrm{u}} \right)\right)\right)$ and $\boldsymbol{l}\left({\boldsymbol{\chi} _i^R}\right)$ and, hence, an SBS can directly replace those different pixels \cite{yang2002computation}. 
In our model, the VR sensors can accurately collect the VR users' movement \cite{HTCsupport}. Moreover, if the SBSs receive accurate tracking data, they can accurately extract each user's location and orientation. Hence, the tracking accuracy depends only on the data error that is caused by the uplink transmission over the wireless link. To compute the processing delay, we first assume that the total amount of computational resources available at each SBS is $M$ which is equally allocated to the associated users. $M$ represents the number of bits that can be processed by each SBS which is determined by each SBS's central processing unit (CPU). Then, the processing delay can be given by:
\begin{equation}\label{eq:dT}
D_i^\textrm{P}\left( {K_i\left(\boldsymbol{s}_{ij}^\textrm{u}\right) } \right)= {\frac{\upsilon\left({{\boldsymbol{l}\left({\boldsymbol{\chi} _i}\left( \boldsymbol{ \gamma}_i^\textrm{u}\left({\boldsymbol{s}_{ij}^\textrm{u}} \right)\right)\right), \boldsymbol{l}\left({\boldsymbol{\chi} _i^R}\right)} } \right)}{{{M \mathord{\left/
 {\vphantom {M {{N_j}}}} \right.
 \kern-\nulldelimiterspace} {{N_j}}}}}},
\end{equation}
 where $0 \le \upsilon\left({{\boldsymbol{l}\left({\boldsymbol{\chi} _i}\left( \boldsymbol{ \gamma}_i^\textrm{u}\left({\boldsymbol{s}_{ij}^\textrm{u}} \right)\right)\right), \boldsymbol{l}\left({\boldsymbol{\chi} _i^R}\right)} } \right)\le L$ represents the number of bits that must be changed when SBS $j$ transforms the VR image from $\boldsymbol{l}\left({\boldsymbol{\chi} _i}\left( \boldsymbol{ \gamma}_i^\textrm{u}\left({\boldsymbol{s}_{ij}^\textrm{u}} \right)\right)\right)$ to $\boldsymbol{l}\left({\boldsymbol{\chi} _i^R}\right)$. $v\left(  \cdot \right)$ depends on the image size, the number of bits used to store a pixel, and the content of the VR image and is the result of the motion search algorithm. Here, we adopt the motion search algorithm of \cite{wei2017qos} to determine $\upsilon\left({{\boldsymbol{l}\left({\boldsymbol{\chi} _i}\left( \boldsymbol{ \gamma}_i^\textrm{u}\left({\boldsymbol{s}_{ij}^\textrm{u}} \right)\right)\right), \boldsymbol{l}\left({\boldsymbol{\chi} _i^R}\right)} } \right)$. {When the SINR ${\boldsymbol{\chi} _i}\left( \boldsymbol{ \gamma}_i^\textrm{u}\left({\boldsymbol{s}_{ij}^\textrm{u}} \right)\right)$ of user $i$ increases, the bit errors in ${\boldsymbol{\chi} _i}\left( \boldsymbol{ \gamma}_i^\textrm{u}\left({\boldsymbol{s}_{ij}^\textrm{u}} \right)\right)$ will decrease. In consequence, $\upsilon\left({{\boldsymbol{l}\left({\boldsymbol{\chi} _i}\left( \boldsymbol{ \gamma}_i^\textrm{u}\left({\boldsymbol{s}_{ij}^\textrm{u}} \right)\right)\right), \boldsymbol{l}\left({\boldsymbol{\chi} _i^R}\right)} } \right)$ will decrease and, hence, the processing delay decreases. Here, the SINR ${\boldsymbol{\chi} _i}\left( \boldsymbol{ \gamma}_i^\textrm{u}\left({\boldsymbol{s}_{ij}^\textrm{u}} \right)\right)$ of user $i$ depends on the resource blocks ${\boldsymbol{s}_{ij}^\textrm{u}} $ allocated to user $i$. } When the deviation between $\boldsymbol{l}\left({\boldsymbol{\chi} _i}\left( \boldsymbol{ \gamma}_i^\textrm{u}\left({\boldsymbol{s}_{ij}^\textrm{u}} \right)\right)\right)$ and $\boldsymbol{l}\left({\boldsymbol{\chi} _i^R}\right)$ increases, more data is needed to correct the image. $ \upsilon\left({{\boldsymbol{l}\left({\boldsymbol{\chi} _i}\left( \boldsymbol{ \gamma}_i^\textrm{u}\left({\boldsymbol{s}_{ij}^\textrm{u}} \right)\right)\right), \boldsymbol{l}\left({\boldsymbol{\chi} _i^R}\right)} } \right) \le L$ is a constraint that captures the maximum number of bits that must be corrected. $N_j$ is the number of the users associated with SBS $j$ and $\frac{M}{N_j}$ is the computation resources allocated to any user $i$'s VR session. From (\ref{eq:dT}), we can see that the processing delay depends on the tracking accuracy, the number of the users associated with SBS $j$, and the resource blocks allocated to user $i$.} The total delay of each user $i$ will hence be:
 \begin{equation}\label{eq:totaldelay}
 D_i\left(\boldsymbol{s}_{ij}^\textrm{d},\boldsymbol{s}_{ij}^\textrm{u}\right)=D_i^\textrm{P}\left( {K_i\left(\boldsymbol{s}_{ij}^\textrm{u}\right) } \right)+D_i^\textrm{T}\left(\boldsymbol{s}_{ij}^\textrm{d},\boldsymbol{s}_{ij}^\textrm{u}\right).
 \end{equation} 

\subsection{Utility Function Model}
Next, we use the framework multi-attribute utility theory \cite{abbas2010constructing} to define a utility function that can effectively capture VR delay and tracking. Using multi-attribute utility theory, we construct a total utility function that jointly considers the delay and tracking of the VR QoS. Conventional techniques for defining a utility function, such as by directly summing up delay and tracking are only valid under the assumption that delay and tracking are independent and {that their relationship is linear \cite[Chapters 2 and 3]{montgomery2015introduction}. However, for VR, the relationship between the delay and tracking is not necessarily linear nor independent. Therefore, we use multi-attribute utility theory \cite{abbas2010constructing} to define the utility function.} The defined utility function can assign to each delay and tracking components of the VR QoS a unique utility value without any constraints. 

In order to construct the total utility function, we first define a conditional utility function for the delay component of the VR QoS \cite{abbas2010constructing}. In the total utility function, both tracking and delay will contribute to the utility value. In contrast, in the conditional utility function of delay, only delay contributes to the utility function and the tracking value is assumed to be given. The method that uses a conditional utility function  to derive the total utility function is analogous to the approach used in probability theory to derive a joint probability distribution from conditional probability distributions \cite{mood1950introduction}. For VR, the total utility function of user $i$ is given by $U_i\left( {{D_i}\left( {{\boldsymbol{s}_{ij}^\textrm{d},\boldsymbol{s}_{ij}^\textrm{u}}} \right),{K_i}\left( {\boldsymbol{s}_{ij}^\textrm{u}} \right)} \right)$. The conditional utility function of delay, $U_i\left( {{D_i}\left( {{\boldsymbol{s}_{ij}^\textrm{d},\boldsymbol{s}_{ij}^\textrm{u}}} \right)\left| {{K_i}\left( {\boldsymbol{s}_{ij}^\textrm{u}} \right)} \right.} \right)$, represents the total utility function given a certain value of the tracking accuracy. Based on \cite{abbas2010constructing}, the conditional utility function of delay for user $i$ can be given by:
 \begin{equation}\label{eq:DFF}
 {
\begin{split}
U_i\left( {{D_i}\!\left( {{\boldsymbol{s}_{ij}^\textrm{d},{{\boldsymbol{s}_{ij}^\textrm{u}}}}} \right)\left| {{K_i}\left( {\boldsymbol{s}_{ij}^\textrm{u}} \right)} \right.} \right)&\\
&\!\!\!\!\!\!\!\!\!\!\!\!\!\!\!\!\!\!\!\!\!\!\!\!\!\!\!\!\!\!\!\!\!\!\!\!\!\!\!\!\!\!\!\!\!\!\!\!\!\!\!\!\!\!\!\!\!\!\!\!\!\!\!\!\!\!\!\!=
\left\{ {\begin{array}{*{20}{c}}
{{\!\!\! \frac{{D_{\max ,i}\left( \!{{K_i}\left(\! {\boldsymbol{s}_{ij}^\textrm{u}} \right)}\right) - {D_i}\left( {{\boldsymbol{s}_{ij}^\textrm{d},{{K_i}\left( {\boldsymbol{s}_{ij}^\textrm{u}} \right)}}} \right)}}{{D_{\max ,i}\!\left( {{K_i}\left( {\boldsymbol{s}_{ij}^\textrm{u}} \right)}\right) - \gamma_{D}}}}, {D_i}\!\left( {{\boldsymbol{s}_{ij}^\textrm{d},{{K_i}\left( {\boldsymbol{s}_{ij}^\textrm{u}} \right)}}} \right) \ge\gamma_{D},}\\
{\;\;\;\;\;\;\;\;\;\;\;\;\;\;\;\;\;\;\;\;\;1,\;\;\;\;\;\;\;\;\;\;\;\;\;\;\;\;\;\;\;\;{D_i}\left( {{\boldsymbol{s}_{ij}^\textrm{d},{{K_i}\left( {\boldsymbol{s}_{ij}^\textrm{u}} \right)}}} \right) <\gamma_{D},}
\end{array}} \right.
\end{split}}
\end{equation}
where $\gamma_{D}$ is the maximal tolerable delay for each VR user (maximum supported by the VR system being used) and $D_{\max ,i}\left( \boldsymbol{s}_{ij}^\textrm{u}\right)=\mathop {\max }\limits_{\boldsymbol{s}_{ij}^\textrm{d}} \left( D_i\left( \boldsymbol{s}_{ij}^\textrm{d},{ {\boldsymbol{s}_{ij}^\textrm{u}} } \right)  \right)$ is the maximum delay of VR user $i$ given $\boldsymbol{s}_{ij}^\textrm{u}$. Here, $U_i\!\left(\! {{D_{\max,i}\left(\boldsymbol{s}_{ij}^\textrm{u}\right)}\!\left| {{K_i}\left( \boldsymbol{s}_{ij}^\textrm{u} \right)} \right.} \!\right)\!=\!0$ and $U_i\!\left( {\gamma_{D}\!\left| {{K_i}\!\left( \boldsymbol{s}_{ij}^\textrm{u} \right)} \right.} \!\right)\!=\!1$. {From (\ref{eq:DFF}), we can see that, when ${D_i}\left( {{\boldsymbol{s}_{ij}^\textrm{d},\boldsymbol{s}_{ij}^\textrm{u}}} \right) <\gamma_{D}$, the conditional utility value will remain 1. This is due to the fact that, when delay ${D_i}\left( {{\boldsymbol{s}_{ij}^\textrm{d},\boldsymbol{s}_{ij}^\textrm{u}}} \right)$ reaches the delay requirement $\gamma_D$, the utility value will reach its maximum of $1$.} Since delay and tracking are both dominant components, we can construct the total utility function, $U_i\left( {{D_i}\left( {{\boldsymbol{s}_{ij}^\textrm{d},\boldsymbol{s}_{ij}^\textrm{u}}} \right),{K_i}\left( \boldsymbol{s}_{ij}^\textrm{u} \right)} \right)$, that jointly considers the delay and tracking based on \cite{abbas2010constructing}. Here, a \emph{dominant component} represents the component that will minimize the total utility function regardless of the value of other components. For VR QoS, delay and tracking are both dominant components. For example, the VR QoS will be minimized when the value of the delay function is at a minimum regardless of the value of tracking accuracy. 
The use of dominant components such as delay and tracking will simplify the formulation of the total utility function \cite{abbas2010constructing}. Therefore, the total utility function of tracking and delay will be \cite{abbas2010constructing}:
\begin{equation}\small\label{eq:Ut}
 {
\begin{split}
&U_i\!\left(\! {{D_i}\!\left(\! {{\boldsymbol{s}_{ij}^\textrm{d},\boldsymbol{s}_{ij}^\textrm{u}}}\! \right)\!,{K_i}\!\left(\! \boldsymbol{s}_{ij}^\textrm{u}\! \right)} \!\right)\!=\! U_i\!\!\left(\! {{D_i}\!\left( \!{{\boldsymbol{s}_{ij}^\textrm{d},\boldsymbol{s}_{ij}^\textrm{u}}} \!\right)\!\left| {{K_i}\!\left(\! {\boldsymbol{s}_{ij}^\textrm{u}}\! \right)} \right.} \!\right)\!U_i\!\left(\! {{K_i}\!\left( \boldsymbol{s}_{ij}^\textrm{u} \right)}\! \right),\\
&\!\!\!\mathop  = \limits^{\left( a \right)} \!\!\left(\!\!1\!-\!\frac{{\left\| {{\boldsymbol{\chi} _i}\left(\boldsymbol{ \gamma}_i^\textrm{u}\left({\boldsymbol{s}_{ij}^\textrm{u}} \right)\right) - {\boldsymbol{\chi} _i^R}} \right\|}}{{\mathop {\max }\limits_{\boldsymbol{s}_{ij}^\textrm{u}}\! \left\| {{\boldsymbol{\chi} _i}\!\!\left( \boldsymbol{ \gamma}_i^\textrm{u}\!\left({\boldsymbol{s}_{ij}^\textrm{u}}\right)\! \right) \!\!-\! \!{\boldsymbol{\chi} _i^R}} \right\|}}\!\!\right)\!\! \frac{{D_{\max ,i}\!\left(\! {{K_i}\!\left( \!{\boldsymbol{s}_{ij}^\textrm{u}} \!\right)}\!\right) \!\!-\! \!{D_i}\!\left( \!{{\boldsymbol{s}_{ij}^\textrm{d},\!{{K_i}\!\left( \!{\boldsymbol{s}_{ij}^\textrm{u}} \right)}}}\! \right)}}{{D_{\max ,i}\!\left( {{K_i}\left( {\boldsymbol{s}_{ij}^\textrm{u}} \right)}\right) - \gamma_{D}}}.
\end{split}}
\end{equation}
where {$U_i\left( {{K_i}\left( \boldsymbol{s}_{ij}^\textrm{u} \right)} \right)$ is the utility function of tracking accuracy.} (a) is obtained by substituting (\ref{eq:Ri}) and (\ref{eq:DFF}) into (\ref{eq:Ut}). From (\ref{eq:Ut}), we can see that, the vector of resource blocks allocated to user $i$ for data transmission, $\boldsymbol{s}_{ij}^\textrm{d}$, and the resource blocks allocated to user $i$ for obtaining the tracking information, $\boldsymbol{s}_{ij}^\textrm{u}$, jointly determine the value of the total utility function. Moreover, this total utility function can assign a unique value to each tracking and delay component of the VR QoS. 

\subsection{Problem Formulation}
 Our goal is to develop an effective resource allocation scheme that allocates resource blocks so as to maximize the users' VR QoS. {This maximization jointly considers the coupled problems of user association, uplink resource allocation, and downlink resource allocation. {This optimization problem can formalized as follows:
 \addtocounter{equation}{0}
\begin{equation}\label{eq:max}
\mathop {\max }\limits_{\boldsymbol{s}_{ij}^\textrm{d},\boldsymbol{s}_{ij}^\textrm{u},\mathcal{U}_j}\sum\limits_{t= 1}^{{T}}\sum\limits_{j \in \mathcal{B}}\sum\limits_{i \in \mathcal{U}_j }{ U_{it}\left( {{D_i}\left( {{\boldsymbol{s}_{ij}^\textrm{d},\boldsymbol{s}_{ij}^\textrm{u}}} \right),{K_i}\left( \boldsymbol{s}_{ij}^\textrm{u} \right)} \right)}
\end{equation}
\begin{align}\label{c1}
\setlength{\abovedisplayskip}{-20 pt}
\setlength{\belowdisplayskip}{-20 pt}
&\rm{s.\;t.}\scalebox{1}{$\;\;\;\;{\left| \mathcal{U}_j \right|} \leqslant  {V_j}, \;\;\;\;\; \forall j \in \mathcal{B},$}\tag{\theequation a}\\
&\scalebox{1}{$\;\;\;\;\;\;\;\;\;\; {{s_{ij,k}^\textrm{d}}}  \in \left\{0,1\right\},\;\;\;\forall i \in \mathcal{U},\forall j \in \mathcal{B}, $} \tag{\theequation b}\\
&\scalebox{1}{$\;\;\;\;\;\;\;\;\;\; {{s_{ij,k}^\textrm{u}}}  \in \left\{0,1\right\},\;\;\;\forall i \in \mathcal{U},\forall j \in \mathcal{B},$} \tag{\theequation c}
\end{align}
where $\mathcal{U}_j $ is the set of the VR users associated with SBS $j$, $\left| \mathcal{U}_j \right|$ is the number of the VR users associated with SBS $j$, ${V_j}$ is the number of the VR users located in the coverage of SBS $j$. $U_{it}$ is the utility value of user $i$ at time $t$. (1a) indicates that the number of VR users associated with SBS $j$ must not exceed the number of VR users located in the coverage of SBS $j$. (1b) and (1c) indicate that each downlink resource block ${s_{ij,k}^\textrm{d}}$ or uplink resource block ${s_{ij,k}^\textrm{u}}$ can be allocated to one VR user.}

In (\ref{eq:max}), the VR QoS of each SBS depends not only on its resource allocation scheme but also on the resource allocation decisions of other SBSs. Consequently, the use of centralized optimization for such a complex problem is not possible as it is largely intractable and yields significant overhead.} To overcome this challenge, we formulate a noncooperative game $\mathcal{G} = \left[ {\mathcal{B},\left\{ {{\mathcal{A}_j}} \right\}_{j \in \mathcal{B}},\left\{ {{u_j}} \right\}}_{j \in \mathcal{B}} \right]$ between the SBSs that can be implemented in a distributed way. Each player $j$ has a set ${\mathcal{A}_j} = \left\{ {\boldsymbol{a}_{j1}, \ldots, \boldsymbol{a}_{j{\left| {{\mathcal{A}_j}} \right|}}} \right\}$ of $\left| {{\mathcal{A}_j}} \right|$ actions. In this game, each action of SBS $j$, ${\boldsymbol{a}_j} = \left( {{\boldsymbol{d}_j},{\boldsymbol{v}_j}} \right)$ consists of: (i) downlink resource allocation vector, ${ \boldsymbol{d}_j} = \left[  {\boldsymbol{s}_{1j}^\textrm{d}, \ldots, \boldsymbol{s}_{V_jj}^\textrm{d}} \right]$. $ \boldsymbol{s}_{ij}^\textrm{d}=\left[s_{ij1}, \ldots, s_{ijS^\textrm{d}}\right]$ represents the vector of downlink resource blocks that SBS $j$ allocates to user $i$ and $s_{ijk} \in \left\{ 1,0\right\}$ with $\sum\limits_{i = 1}^{{V_j}} {s_{ijk}^d}=1$. $s_{ijk}=1$ indicating that channel $k$ is allocated to user $i$ and $s_{ijk}=0$ otherwise. $V_j$ is the number of all users in the coverage area of SBS $j$, and (ii) uplink resource allocation vector, ${ \boldsymbol{v}_j} = \left[  {\boldsymbol{s}_{1j}^\textrm{u}, \ldots, \boldsymbol{s}_{V_jj}^\textrm{u}} \right]$ with $\sum\limits_{i = 1}^{{V_j}} {s_{ijk}^\textrm{u}}=1$. ${\boldsymbol{a}} = \left( {{\boldsymbol{a}_{1}},{\boldsymbol{a}_{2}}, \ldots ,{\boldsymbol{a}_{B}}} \right) \in \mathcal{A}$, represents the action profile of all players and $ \mathcal{A} = \prod\nolimits_{j \in \mathcal{B}} {{\mathcal{A}_{j}}}$.
 
To maximize the VR QoS of each user, the utility function of each SBS $j$ can be given by: 
\begin{equation}\label{eq:u}
u_{j} \left(\boldsymbol{a}_{j}, {\boldsymbol{a}_{-j}}\right)=\sum\limits_{i = 1}^{{V_j}}{ U_i\left( {{D_i}\left( {{\boldsymbol{s}_{ij}^\textrm{d},\boldsymbol{s}_{ij}^\textrm{u}}} \right),{K_i}\left( \boldsymbol{s}_{ij}^\textrm{u} \right)} \right)},
\end{equation}                                                                 
where $\boldsymbol{a}_{j} \in {\mathcal{A}_{j}}$ is an action of SBS $j$ and $\boldsymbol{a}_{-j}$ denotes the action profile of all SBSs other than SBS $j$. Let ${\pi _{j,\boldsymbol{a}_{i}}} \!=\!\!\!\mathop {\lim }\limits_{T \to \infty } \!\!\frac{1}{T}\!\!\sum\limits_{t = 1}^T\!\mathds{1}_{\left\{ \boldsymbol{a}_{j}=\boldsymbol{a}_{ji} \right\}}\!\!= \Pr \left( {{\boldsymbol{a}_j} = {\boldsymbol{a}_{ji}}} \right)$ be the probability that SBS $j$ uses action $\boldsymbol{a}_{ji}$. Hence, ${\boldsymbol{\pi} _j}= \left[ {{\pi _{j,\boldsymbol{a}_{j1}}}, \ldots ,{\pi _{j,\boldsymbol{a}_{j\left| {{A_n}} \right|}}}} \right]$ will be a probability distribution for SBS $j$. We assume that the VR transmission is analyzed during a period that consists of $T$ time slots. Therefore, the average value of the utility function can be given by: 
\begin{equation}
\begin{split}
\bar u_{j} \left(\boldsymbol{a}_{j}, {\boldsymbol{a}_{-j}}\right)&=\mathop {\lim }\limits_{T \to \infty } \frac{1}{T}\sum\limits_{t = 1}^T u_{j} \left(\boldsymbol{a}_{j}, {\boldsymbol{a}_{-j}}\right)\\&=\sum\limits_{\boldsymbol{a} \in \mathcal{A}}\left( {{u_j}\left( {{\boldsymbol{a}_j},{\boldsymbol{a}_{ - j}}} \right)} \prod\limits_{k \in \mathcal{B}} {\pi _{k,\boldsymbol{a}_k}}\right).  
 \end{split}
\end{equation}

To solve the proposed resource allocation game, we use the concept of a mixed-strategy NE, formally \cite{han2012game}:
\begin{definition}\emph{(mixed-strategy Nash Equilibrium): A mixed strategy profile ${\boldsymbol{\pi} ^*} = \left( {\boldsymbol{\pi} _1^*, \ldots ,\boldsymbol{\pi} _{B}^*} \right) = \left( {\boldsymbol{\pi} _j^*,\boldsymbol{\pi}_{ - j}^*} \right)$ is a \emph{mixed-strategy Nash equilibrium} if, $\forall j \in \mathcal{B}$ and $\boldsymbol{\pi} _j$, we have:
\begin{equation}\label{eq:mNE}
{\bar u_j}\left( {\boldsymbol{\pi} _j^*,\boldsymbol{\pi} _{ - j}^*} \right) \ge {\bar u_j}\left( {{\boldsymbol{\pi} _j},\boldsymbol{\pi} _{ - j}^*} \right),
\end{equation}
where 
${\bar u_j}\left( {{\boldsymbol{\pi} _j},{\boldsymbol{\pi} _{ - j}}} \right) =\sum\limits_{{\boldsymbol{a}} \in {\mathcal{A}}} {{u_j}\left( {{\boldsymbol{a}}} \right)\prod\limits_{k \in \mathcal{B}} {{\pi _{k,{\boldsymbol{a}_k}}}}}$ is the expected utility of SBS $j$ selecting the mixed strategy $\boldsymbol{\pi} _j$.}  
\end{definition}

The mixed-strategy NE for the SBSs represents a solution at which each SBS $j$ can maximize the average VR QoS for its users, given the actions of its opponents. 

 \section{Echo State Networks for Self-Organizing Resource Allocation} \label{se:al}
Next, we introduce a learning algorithm used to solve the VR game and find its NE. Since we consider both uplink and downlink resource allocation, the number of actions will be much larger than conventional resource allocation scenarios that typically consider only uplink or downlink resource allocation. {Therefore, as the number of actions significantly increases, by using traditional game-theoretic algorithms such as fictitious play \cite{bacci2016game}, each SBS may not be able to collect all of the information used to calculate the average utility function. Moreover, using such conventional game-theoretic and learning techniques, the SBSs will typically need to re-run the entire steps of the algorithm to reach a mixed-strategy NE as the states of the users and network vary. Hence, the delay during the convergence process of such algorithms may not be able to satisfy the QoS requirement of a dynamic VR network.} To satisfy the QoS requirement for the VR transmission of each SBS, we propose a learning algorithm based on the powerful framework of \emph{echo state networks} (ESN) \cite{chen2017machine}. The proposed ESN-based learning algorithm enables each SBS to predict the value of VR QoS that results from each action and, hence, can reach a mixed-strategy NE without having to traverse all actions. Since the proposed algorithm can store the past ESN information, then, it can find an optimal convergence path from the initial state to a mixed-strategy NE. {Compared to our work in \cite{Chen2016Echo} that is based on two echo state network schemes, the proposed algorithm relies on only one echo state network, which reduces complexity.}
 Next, we introduce the components of an ESN-based learning algorithm and, then, we introduce its update process. 

\subsection{ESN Components}\label{se:alA}
The proposed ESN-based learning algorithm consists of five components: a) agents, b) inputs, c) ESN model d), actions, and e) output, defined as follows.

$\bullet$ \emph{Agent}: The agents in our ESN are the SBSs in the set $\mathcal{B}$. 

$\bullet$ \emph{Actions:} Each action $\boldsymbol{a}_j$ of SBS $j$ jointly considers the uplink and downlink resource blocks, 
which is given by:
\begin{equation}
\boldsymbol{a}_{j}=\left( {{\boldsymbol{d}_{j}},{\boldsymbol{v}_{j}}} \right)= {\left[ {\begin{array}{*{20}{c}}
{\boldsymbol{s}_{1j}^{\textrm{d}} \cdots \boldsymbol{s}_{{V_jj}}^{\textrm{d}}, \boldsymbol{s}_{1j}^{\textrm{u}} \cdots \boldsymbol{s}_{{V_jj}}^{\textrm{u}}}
\end{array}} \right]^\mathrm{T}}.
\end{equation}

In order to guarantee that any action always has a non-zero probability to be chosen, the $\varepsilon$-greedy exploration \cite{31} is adopted. This mechanism is responsible for selecting the actions that each SBS will perform during the learning process while harmonizing the tradeoff between exploitation and exploration. Therefore, the probability with which SBS $j$ chooses action $i$ will be given by: 
\begin{equation}\small\label{eq:p}
\pi _{j,\boldsymbol{a}_{j}} \!=\! \left\{ {\begin{array}{*{20}{c}}
{1 - \varepsilon  + \frac{\varepsilon }{{\left| {{\mathcal{A}_j}} \right|}},\;\; \arg\! \mathop {\max }\limits_{{\boldsymbol{a}_j} \in {\mathcal{A}_j}}\! {\hat{u}_{j}}\!\left( {{\boldsymbol{a}_j}} \right),}\\
{\;\;\;\;\;\;\;\frac{\varepsilon }{{\left| {{\mathcal{A}_j}} \right|}},\;\;\;\;\;\;\;\;\text{otherwise},\;\,\,\,\;\;\;\;\;\;\;\;\;}
\end{array}} \right. 
\end{equation}
where ${\hat u_{\tau,j}\left(\boldsymbol{a}_j\right)}=\!\!\!\!\!\sum\limits_{\boldsymbol{a_{-j}} \in \mathcal{A}_{-j}}\!\!\!\!\!\! {{u_j}\left( {{\boldsymbol{a}_{j}},{\boldsymbol{a}_{ - j}}} \right)} {\pi _{-j,\boldsymbol{a}_{-j}}}$ is the expected utility of an SBS $j$ with respect to the actions of its opponents, $\mathcal{A}_{-j} = \prod\nolimits_{k \ne j, k \in \mathcal{B} } {{\mathcal{A}_{k}}}$ is the set of actions other than the action of SBS $j$ and ${\pi _{ - j,{\boldsymbol{a}_{ - j}}}} = \prod\limits_{k \in \mathcal{B},k \ne j,} {\pi _{k,\boldsymbol{a}_{k}}} $. {${\hat u_{\tau,j}\left(\boldsymbol{a}_j\right)}$ is the marginal probability distribution over the action set of SBS $j$. { ${\hat u_{\tau,j}\left(\boldsymbol{a}_j\right)}$ is the average utility function over all SBSs other than SBS $j$ while $\bar u_{j} \left(\boldsymbol{a}_{j}, {\boldsymbol{a}_{-j}}\right)$ is the average utility value over all SBSs.} From (\ref{eq:p}), we can see that each SBS will assign the highest probability, $1 - \varepsilon  + \frac{\varepsilon }{{\left| {{\mathcal{A}_j}} \right|}}$, to the action that results in the maximum utility value, ${\hat u_{\tau,j}}$. For other actions, the SBS will assign the probability $\frac{\varepsilon }{{\left| {{\mathcal{A}_j}} \right|}}$. {The value of $\varepsilon$ determines the convergence speed.} In this case, as each SBS maximizes the utility $ {\hat u}_{\tau,j}$, the average utility $\bar u_j$ reaches maximum. {Note that, (\ref{eq:p}) is used to find the optimal action for each SBS during the training stage. After training, the SBSs will directly select the optimal action that can maximize $\bar u_j$.}   
To capture the gain that stems from the change of the number of resource blocks allocated to each user $i$, we state the following result:
\begin{theorem}\label{th:1}
\emph{The gain of user $i$'s VR QoS due to the change of the number of resource blocks allocated to user $i$ is:}\\
\emph{\romannumeral1) The gain of user $i$ that stems from the change of the number of uplink resource blocks allocated to user $i$, $\Delta U_i^\textrm{u}$, is:
\begin{equation}\label{eq:Uiu}
\begin{split}
&\Delta U_i^\textrm{u}=\frac{\left(\frac{{{ L}}}{{   \mathop {\max }\limits_{\boldsymbol{s}_{ij}^\textrm{d}}\left( {c_{ij}\left(\boldsymbol{s}^\textrm{d}\right)}\right)}}-\frac{{{ L}}}{{{c_{ij}\left(\boldsymbol{s}_{ij}^\textrm{d}\right)}}}\right){{K_i}\left( {\boldsymbol{s}_{ij}^\textrm{u} + \Delta \boldsymbol{s}_{ij}^\textrm{u}} \right)}}{  \mathop {\max }\limits_{\boldsymbol{s}_{ij}^\textrm{d}} { \left( D_{ij}^\textrm{T}\!\!\left( {\boldsymbol{s}^\textrm{d}}, {\boldsymbol{s}_{ij}^\textrm{u} + \Delta \boldsymbol{s}_{ij}^\textrm{u}}  \right) \right)+D_{i}^\textrm{P}\left({\boldsymbol{s}_{ij}^\textrm{u} + \Delta \boldsymbol{s}_{ij}^\textrm{u}}  \right)   -{\gamma _{D}}}}  \\ &
 -\!\!\left(\!\!\frac{{{ L}}}{{   \mathop {\max }\limits_{\boldsymbol{s}_{ij}^\textrm{d}}\left( {c_{ij}\left(\boldsymbol{s}^\textrm{d}\right)}\right)}}\!-\!\!\frac{{{ L}}}{{{c_{ij}\left(\boldsymbol{s}_{ij}^\textrm{d}\right)}}}\!\!\right)\! \frac{{{K_i}\left( {\boldsymbol{s}_{ij}^\textrm{u}} \right)}}{{ \mathop {\max }\limits_{\boldsymbol{s}_{ij}^\textrm{d}}\! \left( D_{ij}^\textrm{T}\!\left( {\boldsymbol{s}^\textrm{d}},{\boldsymbol{s}_{ij}^\textrm{u}} \right) \!\right)\! \!+\!D_{i}^\textrm{P}\!\left({\boldsymbol{s}_{ij}^\textrm{u}} \right) \! \! -\! {\gamma _{D}}}},
\end{split}
\end{equation} }
\emph{\romannumeral2) The gain of user $i$ that stems from the change of number of downlink resource blocks allocated to user $i$, $\Delta U_i^\textrm{d}$, is:
\begin{equation}{
\Delta {U_i^\textrm{d}} = \left\{ {\begin{array}{*{20}{c}}
{\!\!\frac{{{K_i}\left( {\boldsymbol{s}_{ij}^\textrm{u}} \right){ L}}}{{\left( {{D_{\max ,i}}\left( {\boldsymbol{s}_{ij}^\textrm{u}} \right) - {\gamma _{i,D}}} \right){c_{ij}}\left( {\boldsymbol{s}_{ij}^\textrm{d}} \right)}},{c_{ij}}\left( {\Delta \boldsymbol{s}_{ij}^\textrm{d}} \right) \gg {c_{ij}}\left( { \boldsymbol{s}_{ij}^\textrm{d}} \right)},\\
{\!\!\frac{{{K_i}\left( {\boldsymbol{s}_{ij}^\textrm{u}} \right){ L}{c_{ij}}\left( {\Delta \boldsymbol{s}_{ij}^\textrm{d}} \right)}}{{\left( {{D_{\max ,i}}\left( {\boldsymbol{s}_{ij}^\textrm{u}} \right) - {\gamma _{i,D}}} \right){c_{ij}}\left( {\boldsymbol{s}_{ij}^\textrm{d}} \right)}},{c_{ij}}\left( {\Delta \boldsymbol{s}_{ij}^\textrm{d}} \right) \ll {c_{ij}}\left( {\boldsymbol{s}_{ij}^\textrm{d}} \right)},\\
{\!\!\frac{{{K_i}\left( {\boldsymbol{s}_{ij}^\textrm{u}} \right){ L}}}{{\left( {{D_{\max ,i}}\left( {\boldsymbol{s}_{ij}^\textrm{u}} \right) - {\gamma _{i,D}}} \right)}}\!\! \times \!\!\frac{{{c_{ij}}\left( {\Delta \boldsymbol{s}_{ij}^\textrm{d}} \right)}}{{{c_{ij}}{{\left( {\boldsymbol{s}_{ij}^\textrm{d}} \right)}^2} \!\!+{c_{ij}}\left( {\boldsymbol{s}_{ij}^\textrm{d}} \right){c_{ij}}\left( {\Delta \boldsymbol{s}_{ij}^\textrm{d}} \right)}}, \text{else}.}
\end{array}} \right.}
\end{equation}}
\end{theorem}
\begin{proof}
See Appendix A.
\end{proof}

From Theorem \ref{th:1}, we can see that the tracking accuracy, $K_i$, and the number of uplink resource blocks allocated to user $i$, will directly affect the VR QoS gain of user $i$. Therefore, in order to improve the VR QoS of each user, we can either improve the tracking system or increase the number of the resource blocks allocated to each user according to each user's specific state. Moreover, the gain due to increasing the number of downlink resource blocks depends on the values of data rates ${c_{ij}}\left( {\Delta \boldsymbol{s}_{ij}^\textrm{d}} \right)$ and ${c_{ij}}\left( {\boldsymbol{s}_{ij}^\textrm{d}} \right)$. Hence, the proposed learning algorithm needs to choose the optimal resource allocation scheme to maximize the VR users' QoS. 

$\bullet$ \emph{Input:} The ESN input is a vector $\boldsymbol{x}_{\tau,j}=\left[ {x_{1}, \cdots , x_{B}} \right]^{\mathrm{T}}$ where $x_{j}$ represents the index of the probability distribution that SBS $j$ uses at time $\tau$. $\boldsymbol{x}_{\tau,j}$ is then used to estimate the value of ${\boldsymbol{\hat u}_j}$ that captures the average VR QoS of SBS $j$.

 
 $\bullet$ \emph{ESN Model:} For each SBS $j$, the ESN  model is a learning architecture that can find the relationship between the input $\boldsymbol{x}_{\tau,j}$ and output $\boldsymbol{y}_{\tau,j}$, thus building the function between the SBS's probability distribution and the utility value. 
Mathematically, the ESN model consists of the output weight matrix $\boldsymbol{W}_j^{\textrm{out}} \in {\mathbb{R}^{\left| {{\mathcal{A}_j}} \right| \times \left(N_w+B\right)}}$ and the dynamic reservoir containing the input weight matrix $\boldsymbol{W}_j^{\textrm{in}} \in {\mathbb{R}^{N_w \times B}}$, and the recurrent matrix $\boldsymbol{W}_j =\left[ {\begin{array}{*{20}{c}}
{{w_{11}}}&0&0\\
0& \ddots &0\\
0&0&{{w_{{N_w}{N_w}}}}
\end{array}} \right]$ with $N_w$ being the number of the dynamic reservoir units. {To guarantee that the ESN algorithm can predict the utility values, $N_w$ must be larger than the size of the input vector $\boldsymbol{x}_{\tau,j}$ \cite{APractical}}. Here, the dynamic reservoir is used to store historical ESN information that includes input, reservoir state, and output. Note that the historical ESN information can be used to find a fast converging process from the initial state to the mixed-strategy NE. Here, the number of actions for each SBS determines the output weight matrix and recurrent matrix of each ESN. 
Next, we derive the number of the actions of each SBS $j$, $\left| {{\mathcal{A}_j}} \right|$.      

\begin{proposition}\label{pro:1} \emph{Given the number of the downlink and uplink resource blocks, $S^\textrm{d}$ and $S^\textrm{u}$, as well as the users located in the coverage of SBS $j$, $V_j$, the number of actions for each SBS $j$, $\left| {{\mathcal{A}_j}} \right|$, is given by:\\
\begin{equation}\nonumber
\left| {{\mathcal{A}_j}} \right|= {\left( \begin{array}{l}
{S^\textrm{d}} - 1\\
\left| {\mathcal{N}\left( {{V_j}} \right)} \right| - 1
\end{array} \right) \sum\limits_{\boldsymbol{n} \in \mathcal{N}\left( {{V_j}} \right)} {\prod\limits_{i = 1}^{V_{j} - 1} \left( \begin{array}{l}
{n_i}\\
{S^\textrm{u}} - \sum\nolimits_{k = 1}^{i - 1} {{n_i}} 
\end{array} \right) } }, 
\end{equation}
where $\left( \begin{array}{l}
x\\
y
\end{array} \right) = \frac{{x\left( {x - 1} \right) \cdots \left( {x - y + 1} \right)}}{{y\left( {y - 1} \right) \cdots 1}}$ and $\mathcal{N}\left( V_j \right) = \left\{ {\boldsymbol{n}|\sum\limits_{i = 1}^{V_j } {{n_i}}  = S^\textrm{d}, n_i>0} \right\}$ with $\left|\mathcal{N}\left(V_{j}\right)\right|$ being the number of elements in $\mathcal{N}\left(V_{j}\right)$.} 
\end{proposition}
\begin{proof} See Appendix B.
\end{proof}

Based on Proposition \ref{pro:1}, we can determine the matrix size for both $\boldsymbol{W}_j^{\textrm{out}}$ and $\boldsymbol{W}_j$. From Proposition \ref{pro:1}, we can see that, as the number of users increases, the number of actions increases. Moreover, the increasing number of resource blocks will also increase the number of actions. From Proposition \ref{pro:1}, we can also see that the number of actions in the uplink is much larger than the number of actions in the downlink. This is due to the fact that, in the uplink, the interference of each user changes as the resource blocks allocated to each user vary. However, in the downlink, the actions will not affect the interference of each user.        

$\bullet$ \emph{Output:} The output of the ESN-based learning algorithm at time $t$ is a vector of utility values $\boldsymbol{y}_{\tau,j}= \left[ {{y_{\tau,j1}},{y_{\tau,j2}}, \ldots ,{y_{\tau,j\left| \mathcal{A}_j \right|}}} \right]$. Here, $y_{\tau,ji}$ represents the estimated value of utility ${\hat u_{\tau,j}\left(\boldsymbol{a}_{ji}\right)}$.


\subsection{ ESN-Based Learning Algorithm for Resource Allocation}

\begin{table}[!t]\label{tb1}
  \centering
  \caption{
    \vspace*{-0.3em} ESN-based Learning Algorithm for resource Allocation}\vspace*{-0.7em}
    \begin{tabular}{p{3.3in}}
      \hline \vspace*{-0.6em}
      \textbf{Inputs:}\,\, Mixed strategy $\boldsymbol{x}_{\tau,j}$  \vspace*{-0.5em}\\
\hspace*{1em}\textit{Initialize:}   \vspace*{-0.2em}
$\boldsymbol{W}_j^{\textrm{in}}$, $\boldsymbol{W}_j$, $\boldsymbol{W}_j^{\textrm{out}}$, and $\boldsymbol{y}_{j}=0$.
\vspace*{-0.05em}
\hspace*{0em}\begin{itemize}\vspace*{-0.1em}
\item[] \hspace*{0em} \textbf{for} each time $\tau$ \textbf{do}.
\item[] \hspace*{1em}(a) Estimate the value of the utility function based on (\ref{eq:update}).
\item[] \hspace*{0.5em} \textbf{if} $\tau=1$ 
\item[] \hspace*{2em}(b) Set the mixed strategy $\boldsymbol{\pi}_{\tau,j}$ uniformly.
\item[] \hspace*{0.5em} \textbf{else}
\item[] \hspace*{2em}(c) Set the mixed strategy $\boldsymbol{\pi}_{\tau,j}$ based on (\ref{eq:p}).
\item[] \hspace*{0.5em} \textbf{end if}
\item[] \hspace*{1em}(d) Broadcast the index of the mixed strategy to other SBSs.
\item[] \hspace*{1em}(e) Receive the index of the mixed strategy as input $\boldsymbol{x}_{\tau,j}$. 
\item[] \hspace*{1em}(f) Perform an action based on $\boldsymbol{\pi}_{\tau,j}$ and calculate the actual \hspace*{2.2em} utility value. 
\item[] \hspace*{1em}(e) Update the dynamic reservoir state based on (\ref{eq:state}).
\item[] \hspace*{1em}(g) Update the output weight matrix based on (\ref{eq:w}).
\item[] \hspace*{0em} \textbf{end for}
\end{itemize}\vspace*{-0.3cm}\\
   \hline
    \end{tabular}\label{tab:algo}\vspace{-0.5cm}
\end{table}   

We now present the proposed ESN-based learning algorithm to find a mixed strategy NE. 
The proposed learning algorithm can find an optimal convergence path from any initial state to a mixed-strategy NE. In particular, the proposed algorithm enables each SBS to reach a mixed-strategy NE by traversing minimum number of strategies after training. In order to find the optimal convergence path, the proposed algorithm must store the past ESN information that consists of input, reservoir states, and output. The past ESN information from time $0$ up until time $\tau$ is stored by the dynamic reservoir state ${\boldsymbol{\mu  }_{\tau,j}}$.
 The dynamic reservoir state of SBS $j$ at time $\tau$ is: 
 \begin{equation}\label{eq:state}
{\boldsymbol{\mu}_{\tau,j}} ={\mathop{f}\nolimits}\!\left( {\boldsymbol{W}_j{\boldsymbol{\mu}_{\tau - 1,j}} + \boldsymbol{W}_j^\textrm{in}{\boldsymbol{x}_{\tau,j}}} \right),
\end{equation}   
where  $f\!\left(x\right)=\frac{{{e^x} - {e^{ - x}}}}{{{e^x} + {e^{ - x}}}}$ is the tanh function. From (\ref{eq:state}), we can see that the dynamic reservoir state consists of the past dynamic reservoir states and the mixed strategy at time $\tau$. Thus, the dynamic reservoir state actually stores the mixed strategy from time 0 to time $\tau$. Based on the dynamic reservoir state, the proposed ESN algorithm will combine with the output weight matrix to estimate the value of conditional utility value. The estimation of the utility value can be given by: 
\begin{equation}\label{eq:update}
\boldsymbol{y}_{\tau,j} = {\boldsymbol{W}_{\tau,j}^\textrm{out}} \left[ {\begin{array}{*{20}{c}}
  {\boldsymbol{\mu}_{\tau,j}} \\ 
  {\boldsymbol{x}_{\tau,j}} 
\end{array}} \right], 
\end{equation}
where ${\boldsymbol{W}_{\tau,j}^\textrm{out}}$ is the output weight matrix at time slot $\tau$. To enable the ESN to use reservoir state $ {{\boldsymbol{\mu}_{\tau,j}}}$ to predict the utility value, ${\hat u_{\tau,ji}}$, due to action ${\boldsymbol{a}_{ji}}$, we must train the output matrix $\boldsymbol{W}_j^\textrm{out}$ using a linear gradient descent approach, which is:
\begin{equation}\label{eq:w}
{\boldsymbol{W}_{\tau + 1,ji}^\textrm{out}} = {\boldsymbol{W}_{\tau,ji}^\textrm{out}} + {\lambda} \left( {\hat u_{\tau,ji} -y_{\tau,ji} \left(\boldsymbol{x}_{\tau,j}^j, {\boldsymbol{a}_{ji}} \right)} \right){\boldsymbol{\mu}_{\tau,j}^{\mathrm{T}}},
\end{equation}
where ${\boldsymbol{W}_{\tau,ji}^\textrm{out}}$ is row $i$ of ${\boldsymbol{W}_{\tau,j}^\textrm{out}}$, $\lambda$ is the learning rate, and $\hat u_{\tau,ji}$ is the actual utility value. Here, $y_{\tau,ji}$ is the estimation of the utility value $\hat u_{\tau,ji}$. Table~II summarizes our algorithm.

\subsection{Convergence of the ESN-Based Learning Algorithm}
Now, we analyze the convergence of the proposed ESN-based learning algorithm. 
\begin{theorem}\label{th:2} \emph{The proposed ESN-based learning algorithm converges to the utility value, $ \boldsymbol{\hat u}_j$, if any following conditions is satisfied:\\} 
\emph{\romannumeral1) $\lambda$ is a constant and $\mathop {\min }\limits_{\boldsymbol{W}_{ji}^\textrm{in},{\boldsymbol{x}_{\tau ,j}},{\boldsymbol{x}'_{\tau ,j}}} \boldsymbol{W}_{ji}^\textrm{in}\left({\boldsymbol{x}_{\tau,j}}-{\boldsymbol{x}'_{\tau',j}}\right) \ge 2$, where $\boldsymbol{W}_{ji}^\textrm{in}$ represents the row $i$ of $\boldsymbol{W}_{j}^\textrm{in}$.\\
\romannumeral2) $\lambda$ satisfies the Robbins-Monro conditions \cite{26} (${\lambda\left(t\right)} > 0,\sum\nolimits_{t = 0}^\infty  {{\lambda\left(t\right)} = +\infty ,\sum\nolimits_{t = 0}^\infty  {\lambda^2\left(t\right) <+ \infty } } $).}

\end{theorem} 
\begin{proof}See Appendix C.        
\end{proof}

From Theorem \ref{th:2}, we can see that the convergence of the proposed algorithm depends on the values of the input weight matrix $\boldsymbol{W}_{ji}^\textrm{in}$ and the input ${\boldsymbol{x}_{\tau ,j}}$. These values also affect the capacity of the ESN's memory. Here, the memory of a ESN represents the ability that an ESN can store the past ESN information. Therefore, the proposed algorithm of SBS $j$ can converge to the conditional utility function $ \boldsymbol{\hat u}_j$ by choosing appropriate $\boldsymbol{W}_{ji}^\textrm{in}$ and ${\boldsymbol{x}_{\tau ,j}}$. Indeed, the proposed learning algorithm can converge to $ \boldsymbol{\hat u}_j$ even when $\boldsymbol{W}_{ji}^\textrm{in}$ and ${\boldsymbol{x}_{\tau ,j}}$ are generated randomly. 
This is due to the fact that the probability of $\boldsymbol{u}_{\tau,j}=\boldsymbol{u}'_{\tau',j}$ is particularly small since $\boldsymbol{u}_{\tau,j}$ has a larger number of elements (i.e. more than 500).  
{Here, we note that, compared to our work in \cite{Chen2016Echo} that uses two echo state networks to guarantee convergence, the proposed ESN algorithm uses only one echo state network and is still guaranteed to converge as shown in Theorem 2.}
 { Moreover, this new proof of convergence, we can invoke our result in  the fact that the algorithm will reach a mixed NE follows directly from \cite[Theorem 2]{Chen2016Echo} to guarantee that the convergence point will be a mixed NE, as follows.}   

\begin{corollary}\label{th:3}\emph{ The ESN-based learning algorithm of each SBS $j$ converges to a mixed-strategy Nash equilibrium, with the mixed strategy probability ${\boldsymbol{\pi}_j^*}$, $\forall j \in \mathcal{B}$.}
\end{corollary}
Based on (\ref{eq:p}), each SBS will assign the highest probability to the action that results in the maximum value of $\hat u_{j,\max}=\mathop {\max }\limits_{{\boldsymbol{a}_j} \in {\mathcal{A}_j}}\! {\hat{u}_{\tau,j}}\!\left( {{\boldsymbol{a}_j}} \right)$. Therefore, when each SBS $j$ reaches the optimal mixed strategy $\boldsymbol{\pi}_j^*$ and the ESN reaches the maximal utility $\boldsymbol{\hat u}_j$, the maximum value of the average utility $\bar u_j$ can be given by $\left(1 - \varepsilon  + \frac{\varepsilon }{{\left| {{\mathcal{A}_j}} \right|}}\right) \hat u_{j,\max}$. {Since the mixed-strategy NE depends on the utility values that are not unique, then, the resulting mixed-strategy NE is not unique. While characterizing all possible NEs is challenging analytically, in our simulations in Section~IV, we will analyze the performance of different NEs.}  
\subsection{Implementation and Complexity}

The proposed algorithm can be implemented
in a distributed way. At the initial state, the reservoir and output of the ESN will be zero. During each iteration, the output and reservoir of the ESN will be updated based on (\ref{eq:state})-(\ref{eq:update}). Based on the ESN's output, each SBS will update its mixed strategy and broadcast the index of this mixed strategy to other SBSs. {Compared to the size of VR content and tracking information, the size of the index of the strategy is very small (it can be represented with less than 16 bits). Therefore, these interactions between SBSs are independent of the network size and, hence, they incur no notable overhead.} The resulting mixed-strategy NE depends on the utility value resulting from each action of each SBS. 

The objective of our game is to find the equilibrium mixed strategy for each SBS. Hence, the complexity of the proposed algorithm depends on the number of mixed strategies. Based on (\ref{eq:p}), we can see that the number of mixed strategies is equal to the number of actions. Since the worst-case for each SBS is to traverse all actions, the worst-case complexity of the proposed algorithm is $O(\left| {{\mathcal{A}_{1}}}\right| \times \dots \times \left| {{\mathcal{A}_{B}}}\right|)$. However, the worst-case complexity pertains to a rather unlikely scenario in which all SBSs choose their optimal probability strategies after traversing all other mixed strategies during each period $\tau$ and, hence, the probability of occurrence of the worst-case scenario is $\small {\left( {1 - \frac{\varepsilon }{{\left| {{\mathcal{A}_1}} \right|}}} \right)^{\left| {{\mathcal{A}_1}} \right| - 1}} \times  \cdots  \times {\left( {1 - \frac{\varepsilon }{{\left| {{\mathcal{A}_B}} \right|}}} \right)^{\left| {{\mathcal{A}_B}} \right| - 1}}$\footnote{{ Based on (\ref{eq:p}), for each SBS, the probability that the optimal action is not selected at each iteration is $\left( {1 - \frac{\varepsilon }{{\left| {{\mathcal{A}_j}} \right|}}} \right)$and hence, the probability that the optimal action is selected at the last iteration is ${\left( {1 - \frac{\varepsilon }{{\left| {{\mathcal{A}_j}} \right|}}} \right)^{\left| {{\mathcal{A}_j}} \right| - 1}}$. Therefore, the probability of all SBSs select their optimal action at last iteration is $ {\left( {1 - \frac{\varepsilon }{{\left| {{\mathcal{A}_1}} \right|}}} \right)^{\left| {{\mathcal{A}_1}} \right| - 1}} \times  \cdots  \times {\left( {1 - \frac{\varepsilon }{{\left| {{\mathcal{A}_B}} \right|}}} \right)^{\left| {{\mathcal{A}_B}} \right| - 1}}$.  }}. Therefore, the proposed algorithm will converge faster than the worst-case with probability $1-\small {\left( {1 - \frac{\varepsilon }{{\left| {{\mathcal{A}_1}} \right|}}} \right)^{\left| {{\mathcal{A}_1}} \right| - 1}} \times  \cdots  \times {\left( {1 - \frac{\varepsilon }{{\left| {{\mathcal{A}_B}} \right|}}} \right)^{\left| {{\mathcal{A}_B}} \right| - 1}}$. 
Moreover, as the number of SBSs increases, based on Proposition \ref{pro:1}, the number of actions for each SBS significantly decreases. Thus, the worst-case complexity of the proposed algorithm will also decrease. In addition, our ESN-based learning algorithm can use the past information to predict the value of utility $\boldsymbol{\hat u}_{j}$ and, hence, each SBS can obtain the VR QoS of each SBS without implementing all of its possible actions, which will also reduce the complexity of the proposed algorithm. During the convergence process, the proposed algorithm needs to harmonize the tradeoff between exploration and exploitation. {Exploration is used to enable each SBS to explore actions so as to find a better solution.} Exploitation refers to the case in which each SBS will use the current optimal action at this iteration. This tradeoff is controlled by the $\varepsilon$-greedy exploration specificed in (\ref{eq:p}). If we increase the probability of exploration, the proposed algorithm will use more actions and, hence, the number of iterations that the proposed algorithm requires to reach a mixed-strategy NE will decrease. However, increasing the probability of exploration may reduce the VR QoS of each user when the selected action is worse than the current optimal action. 

\section{Simulation Results and Analysis}
For our simulations, we consider an SCN deployed within a circular area with radius $r = 100$ m. $U=25$ users and $B=4$ SBSs are uniformly distributed in this SCN area. The HTC Vive VR device is considered in our simulations and, hence, the number of pixels for a panoramic image is $1920 \times 1080$ and each pixel is stored in 16 bits \cite{timmurphy.org}. The flashed rate which represents the update rate of a VR image, is 60 images per second and the factor of compression is 150 \cite{bacstuug2016towards}. Since two panoramic images consist of one VR image (one panoramic image per eye), the rate requirement for wireless VR transmission will be 25.32 Mbit/s\footnote{{Here, for each second, each SBS needs to transmit $1920\times1080 \times16 \times60 \times2=3,981,312,000$ bits $=3796.875$ Mbits to each user. Since the factor of compression is 150, the rate requirement is ${{3796.875.75} \mathord{\left/
 {\vphantom {{1920} {\left( {8 \times 1024 \times 1024 \times 300} \right)}}} \right.
 \kern-\nulldelimiterspace} {{  150} }}=25.3125$ Mbit/s.}}. {\color{blue}  }
 The bandwidth of each resource block is set to {$10 \times 180$ kHz} \cite{21}. {We use a standard H.264/AVC video codec. The VR video is encoded at 60 frames per second and the encoding rate of each VR video is 4 Mbps.}
The detailed simulation parameters are listed in Table  \uppercase\expandafter{\romannumeral3}. For comparison purposes, we use three baselines: a) Proportional fair algorithm in \cite{kelly1997charging} {that allocates the resource blocks based on the number of the resource blocks needed to maximize each user's utility value}, b) {Q-learning algorithm in \cite{30} that has considered the historical utility value to estimate the future utility value}, and c) Optimal heuristic search algorithm. Note that, in order to compare with the proportional fair algorithm, all SBSs choose the action with highest probability among the mixed strategy when the proposed algorithm reaches a mixed-strategy NE. This is used for all results in which we compare to proportional fair. Here, the users' localization data is measured from actual wired HTC Vive VR devices and the wireless transmission is simulated, in order to compute the tracking accuracy. All statistical results are averaged over a large number of independent runs.

\begin{table}
  \newcommand{\tabincell}[2]{\begin{tabular}{@{}#1@{}}#2\end{tabular}}
\renewcommand\arraystretch{0.6}
 \caption{
    \vspace*{-0.2em}SYSTEM PARAMETERS}\vspace*{-0.8em}
\centering  
\begin{tabular}{|c|c|c|c|c|c|}
\hline
\textbf{Parameter} & \textbf{Value} & \textbf{Parameter} & \textbf{Value}& \textbf{Parameter} & \textbf{Value}  \\
\hline
$F$ & 1000 & $P_B$ & 20 dBm&$\upsilon$&5\\
\hline
$B$ & 4 & $S^\textrm{u}$&  5&$S^\textrm{d}$&5\\
\hline
$N_w$ & 1000 & $\sigma ^2$ & -95 dBm&$r_B$& $25$ m\\
\hline
$ N_{v}$ &6&$\lambda$ & 0.03&A &100 kB \\
\hline
\end{tabular}
 \vspace{-0.5cm}
\end{table}

 \begin{figure}[!t]
  \begin{center}
   \vspace{0cm}
    \includegraphics[width=7cm]{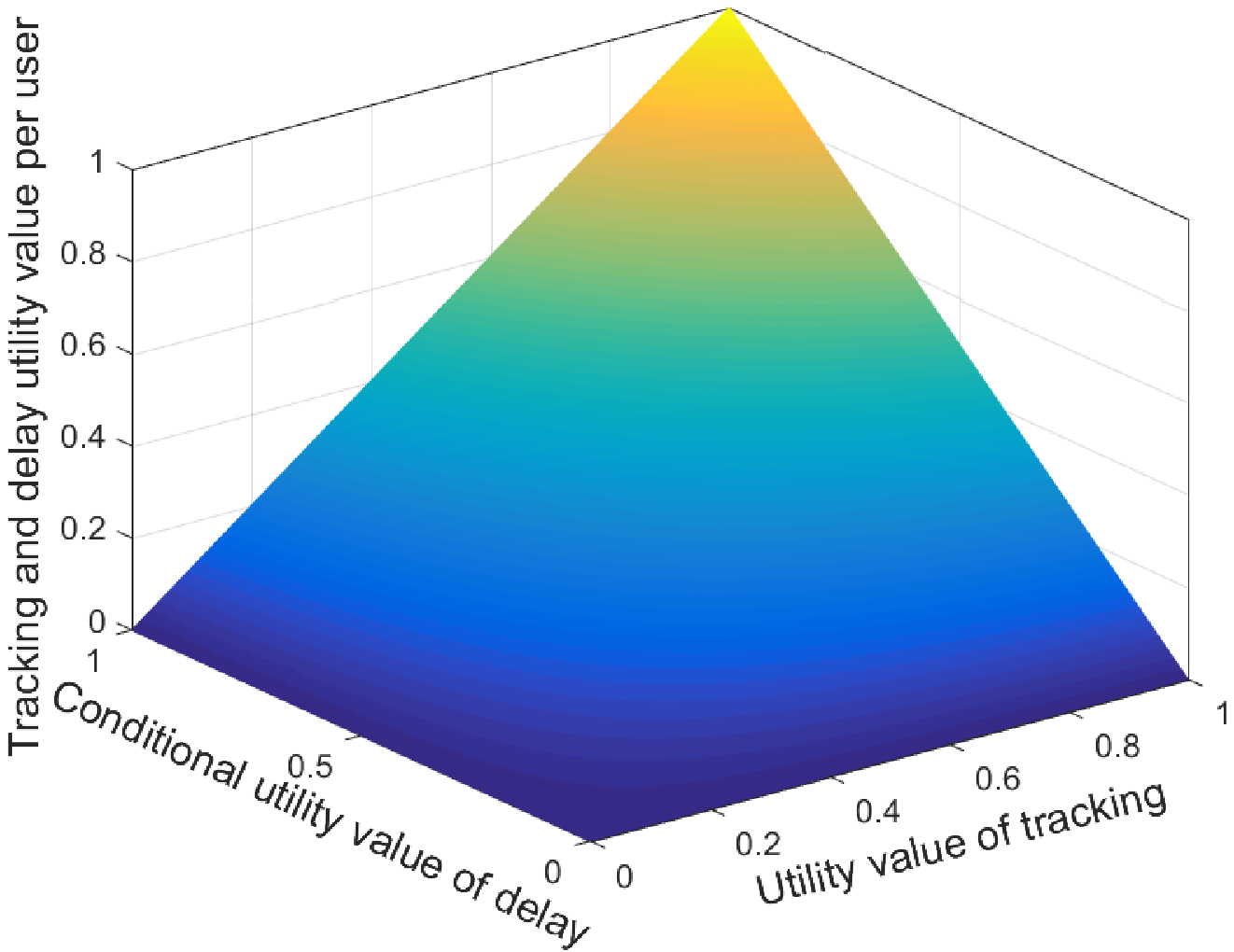}
    \vspace{-0.2cm}
    \caption{\label{figure4}Total VR QoS utility of each user vs. the tracking and delay utilities. Here, total VR QoS utility refers to (\ref{eq:Ut}). }
  \end{center}\vspace{-0.6cm}
\end{figure}
\vspace{-0cm}

Fig. \ref{figure4} shows how each user's VR QoS utility varies as the tracking and delay utilities change. In Fig. \ref{figure4}, different colors indicate different total VR QoS utilities. From Fig. \ref{figure4}, we can see that, when the delay (tracking) utility is 0, the total VR QoS utility will be 0 regardless of the tracking (delay) value. Thus, both tracking accuracy and delay will affect the VR QoS. In Fig. \ref{figure4}, we can also see that only when both tracking and delay utilities are 1, the total VR QoS utility is maximized. This is due to the fact that the multi-attribute utility theory model assigns to each tracking and delay components of the VR QoS a unique value. Clearly, it is clear that, the proposed total utility function can effectively capture the VR QoS.


 \begin{figure}[!t]
  \begin{center}
   \vspace{0cm}
    \includegraphics[width=7cm]{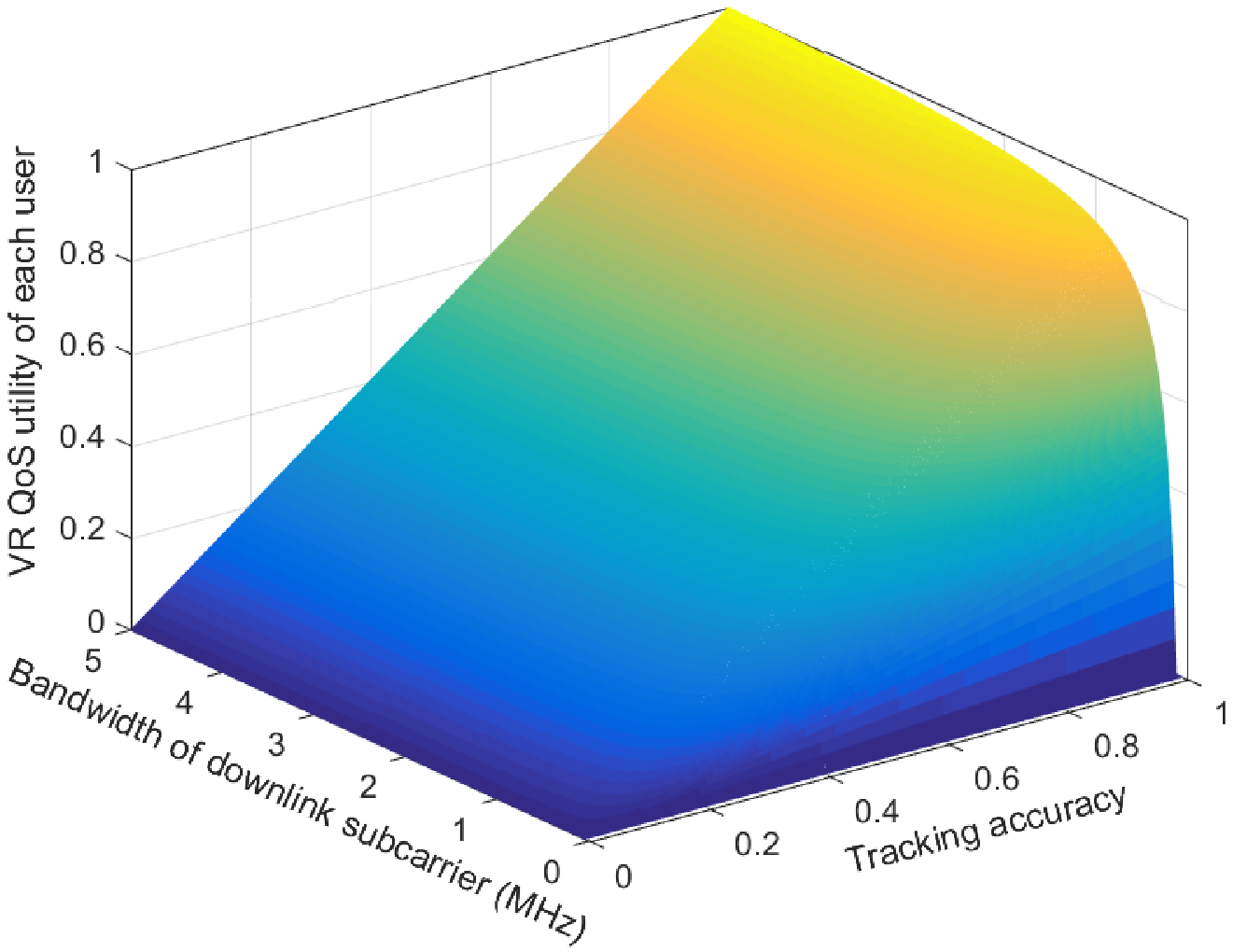}
    \vspace{-0.2cm}
    \caption{\label{figure5}Total VR QoS utility of each user vs. the tracking and delay utilities. Here, total VR QoS utility refers to (\ref{eq:Ut}). }
  \end{center}\vspace{-0.7cm}
\end{figure}
\vspace{-0cm}

{Fig. \ref{figure5} shows how each user's total utility changes as function of the tracking accuracy and the bandwidth of each downlink resource block group (Fig. \ref{figure5} uses a similar color legend as Fig. \ref{figure4}). From Fig. \ref{figure5}, we can see that when the bandwidth of downlink resource block group (tracking accuracy) is 0, the total VR QoS utility is 0 regardless of the tracking accuracy (bandwidth of a downlink resource block group). This is due to the fact that the VR QoS depends on both delay and tracking. This corresponds to a scenario in which SBS $j$ has enough downlink bandwidth to send a VR image to the user while the tracking information is inaccurate. In this case, SBS $j$ cannot construct the accurate VR image due to the inaccuracy of user's localization and, hence, the VR QoS of this user will be 0. Fig. \ref{figure5} also shows that, when the tracking accuracy is 1 and the bandwidth of a downlink resource block group is over 4 MHz, the total VR QoS utility will be maximized. This verifies the result of Theorem 1. }                

 \begin{figure}
\centering
\vspace{0cm}
\subfigure[Average delay utility of each serviced user]{
\label{figure11a} 
\includegraphics[width=7cm]{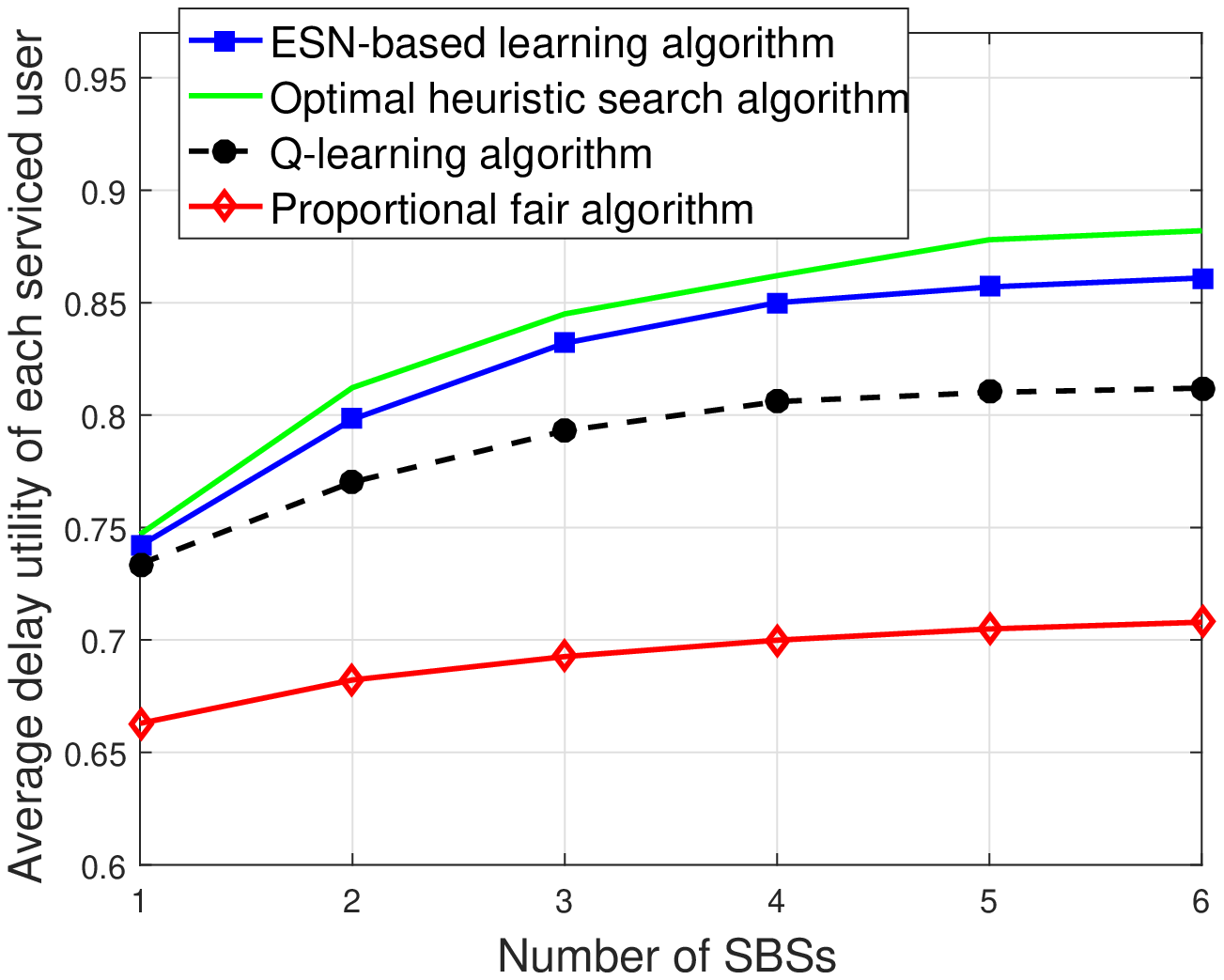}}
\subfigure[Average delay of each serviced user]{
\label{figure11b} 
\includegraphics[width=7cm]{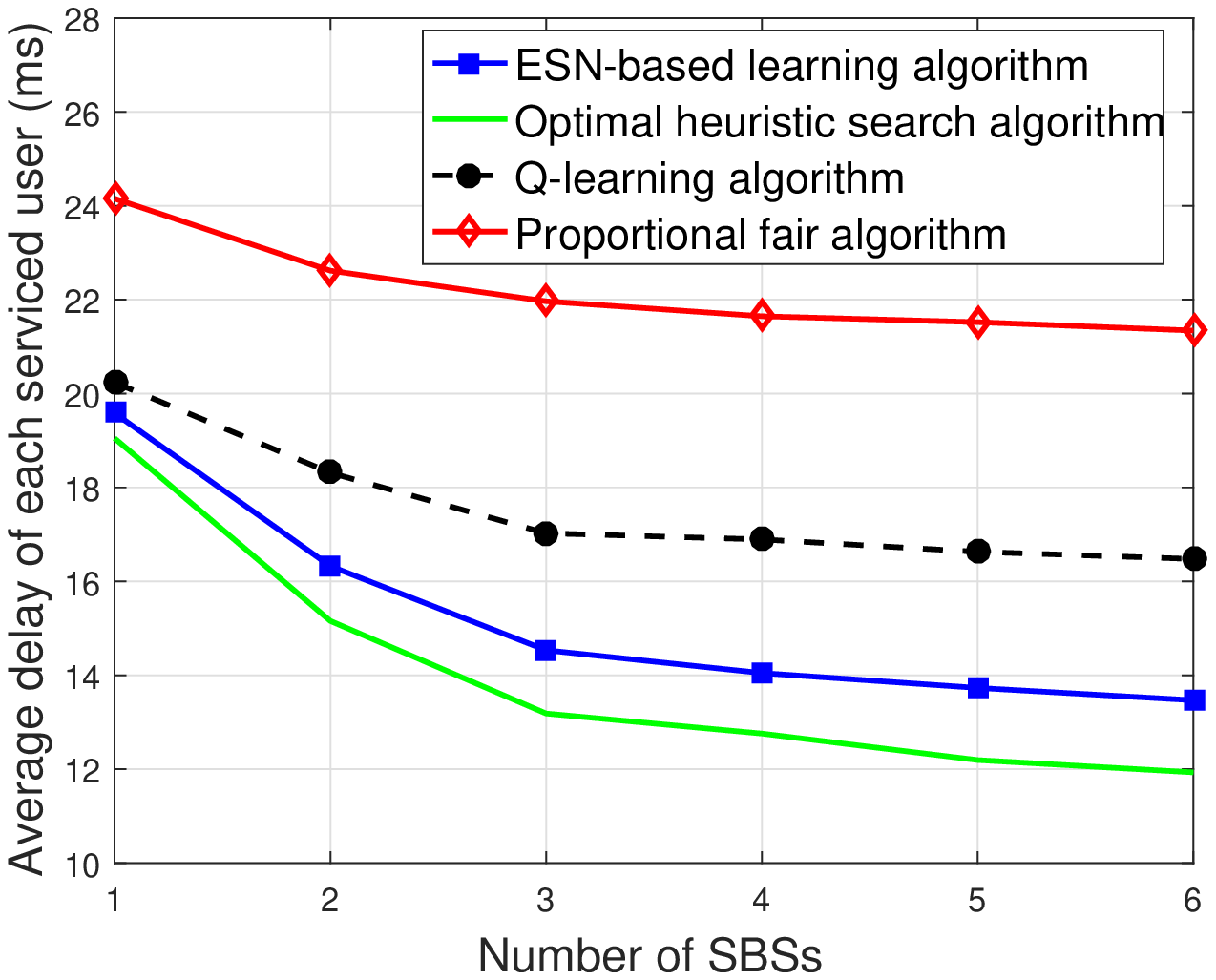}}
  \vspace{-0.2cm}
 \caption{\label{figure11} {Delay for each serviced user vs. the number of SBSs.}}
  \vspace{-0.5cm}
\end{figure}

In Fig. \ref{figure11}, we show how the average delay utility for each user varies as the number of SBSs increases. From Figs. \ref{figure11a} and \ref{figure11b}, we can see that, as the number of SBSs increases,  the average delay utility for each serviced user increases and the transmission delay of each user decreases.      
This is due to the fact that as the number of SBSs increases, the number of users located in each SBS's coverage decreases and, hence, the average delay utility increases. However, as the number of SBSs keeps increasing, the average delay utility increases slowly. This stems from the fact that the interference from the SBSs to the users increases as the number of SBSs continues to increase. Fig. \ref{figure11b} also shows that the proposed algorithm  achieves up to {19.6\%} gain in terms of average delay compared to the Q-learning algorithm for the case with 6 SBSs. In Fig. \ref{figure11b}, {we can also see that the proposed ESN-based learning algorithm enables the wireless VR transmission to meet typical delay requirement of VR applications that consists of both the transmission delay and processing delay (typically 20 ms \cite{WhatVRMichael})}. These gains stem from the fact that the proposed algorithm uses the past ESN information stored at the ESN model to find a better solution for the proposed game.   

 \begin{figure}[!t]
  \begin{center}
   \vspace{0cm}
    \includegraphics[width=7cm]{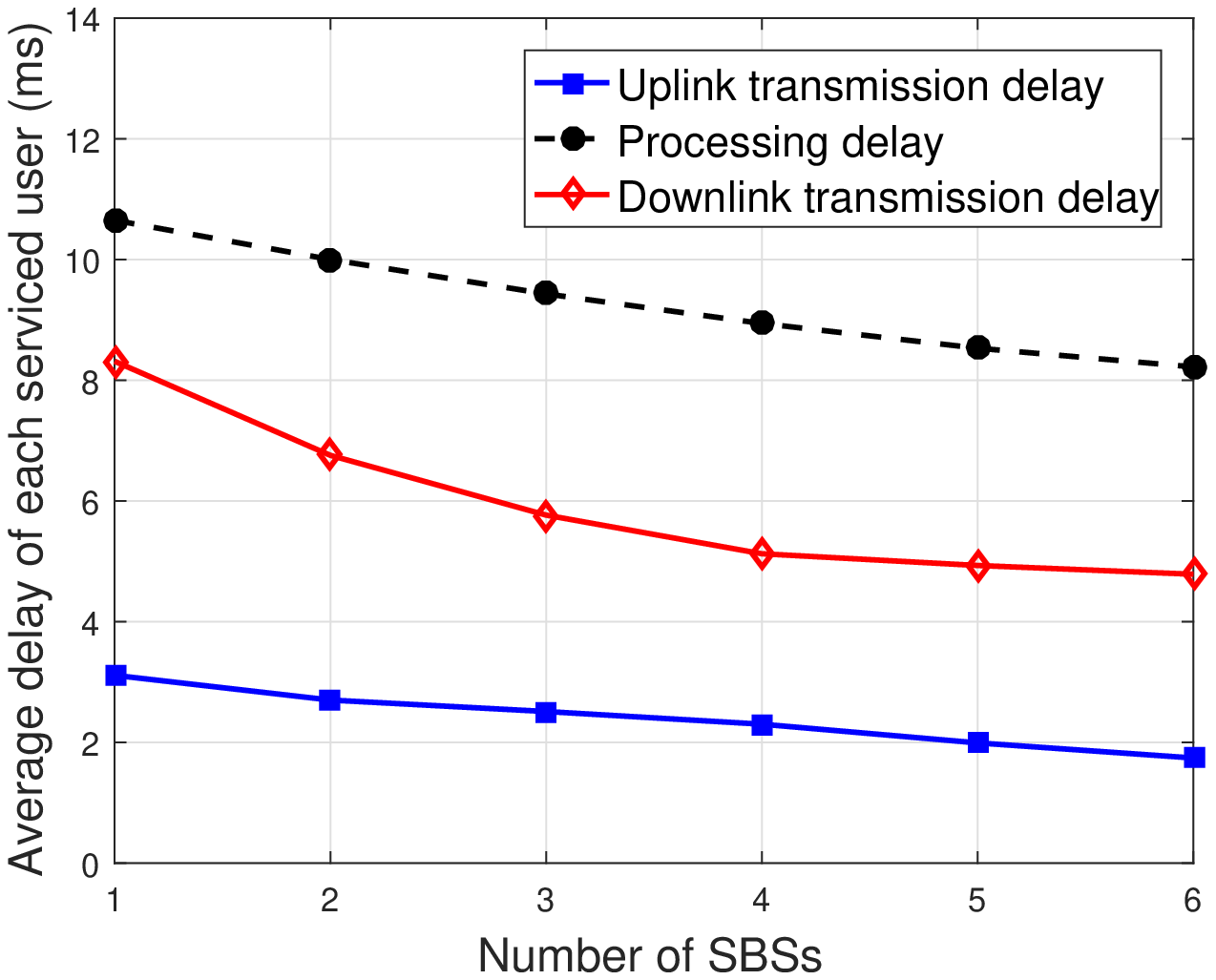}
    \vspace{-0.2cm}
    \caption{{\label{delay}Delay for each serviced user vs. the number of SBSs.}}
  \end{center}
    \vspace{-0.8cm}
\end{figure}

{Fig. \ref{delay} shows how the transmission and processing delays change as the number of the SBSs varies. From Fig. \ref{delay}, we can see that all delay components of each user decrease as the number of SBSs increases. This is due to the fact that, as the number of SBSs increases, the users will have more SBS choices and the distances from the SBSs to the users decrease, and, hence, the SINR and the potential resources (resource blocks) allocated to each user will increase.} 

{
 \begin{figure}[!t]
  \begin{center}
   \vspace{0cm}
    \includegraphics[width=7cm]{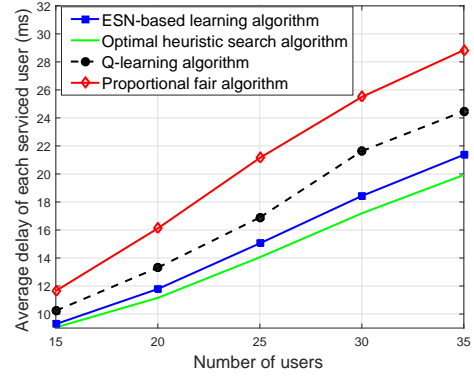}
    \vspace{-0.3cm}
    {\caption{{\label{figure19}Average delay for each serviced user as the number of users varies.}}}
  \end{center}
    \vspace{-0.7cm}
\end{figure}

In Fig. \ref{figure19}, we show how the average delay for each serviced user changes as the number of VR users varies. From Fig. \ref{figure19}, we can see that, as the number of VR users increases, the average delay of each VR user increases. {This is due to the fact that, as the number of users increases, the number of resource blocks that is allocated to each VR user decreases.} Fig. \ref{figure19} also shows that the deviation between the proposed algorithm and Q-learning increases as the number of VR users increases.
This is due to the fact that, as the number of users increases, the number of the users associated with each SBS increases. In consequence, the number of actions for each SBS increases and, hence, the SBSs need to record more QoS values resulting from these actions. Compared to Q-learning that uses a matrix to record the QoS values, the proposed algorithm uses an ESN to approximate the function of QoS values and, hence, the proposed algorithm can record more QoS utility values compared to Q-learning. Fig. \ref{figure19} also shows that the proposed algorithm can yield 29.8\% gain of the average delay compared to proportional fair algorithm. This gain stems from the fact that the proposed algorithm can find the relationship between resource block allocation strategies and utility values so as to maximize the utility values.}

 \begin{figure}[!t]
  \begin{center}
   \vspace{0cm}
    \includegraphics[width=7cm]{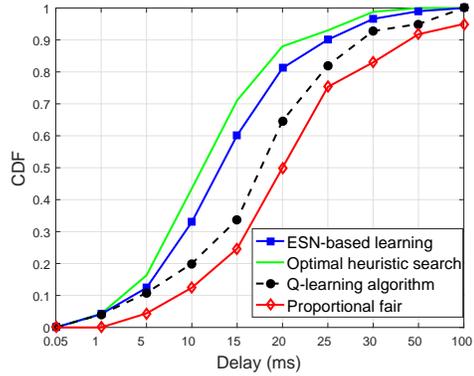}
    \vspace{-0.3cm}
    \caption{\label{cdf}{CDFs of the delay resulting from the different algorithms. }}
  \end{center}\vspace{-0.7cm}
\end{figure}
\vspace{-0cm}
 {Fig. \ref{cdf} shows the cumulative distribution function (CDF) for the total delay resulting from all of the considered schemes. In Fig. \ref{cdf}, we can see that, the delay of 100\% of users resulting from all of the considered algorithms will be above 0.05 ms. This is due to the fact that, the delay requirement of each user is higher than 0.05 ms and, hence, when the user's delay requirement is satisfied, the SBSs will allocate the resource blocks to other users. Fig. \ref{cdf} also shows that the proposed approach improves the CDF of up to 25\% and 50\% gains at a delay of 20 ms compared to Q-learning and proportional fair algorithm. These gains stem from the fact that the proposed algorithm can estimate the VR QoS resulting from each SBS's actions accurately and, hence, can find a better solution compared to Q-learning and proportional fair algorithms.} {Fig. \ref{cdf} also shows that there exists a delay variance for users who achieve the $20$ ms delay target. If needed, the network can reduce this variance by adjusting the size of the resource blocks.}

 \begin{figure}
\centering
\vspace{0cm}
\subfigure[Average total VR QoS utility for all users.]{
\label{fig9a} 
\includegraphics[width=7cm]{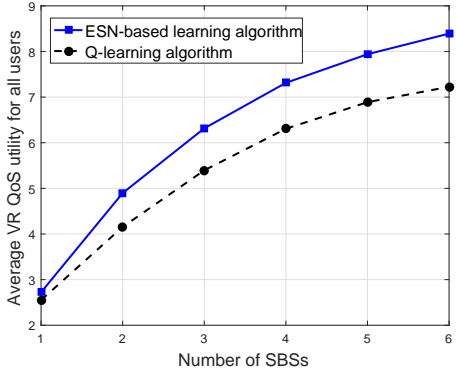}}
\subfigure[Total VR QoS utility for all users]{
\label{fig9b} 
\includegraphics[width=7cm]{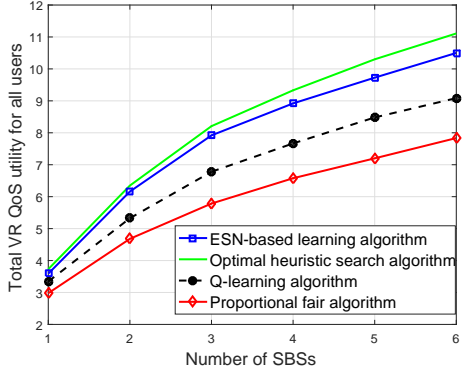}}
  \vspace{-0.2cm}
 \caption{\label{figure9} {VR QoS for all users vs. the number of SBSs. Here, total VR QoS utility refers to (\ref{eq:u})}. {The proportional fair algorithm and optimal heuristic search do not use mixed strategies and, hence, they are not shown in Fig. \ref{fig9a}.}}
  \vspace{-0.7cm}
\end{figure}

Fig. \ref{figure9} shows how the VR QoS for all users changes as the number of SBSs varies. 
 From Figs. \ref{fig9a} and \ref{fig9b}, we can see that both total utility values and average total utility values (at the mixed-strategy NE) of all considered algorithms increase as the number of SBSs increases. This is due to the fact that, as the number of SBSs increases, the number of users located within the coverage of each SBS increases and the distances from the SBSs to their associated users decrease. Fig. \ref{fig9a} shows that the proposed algorithm can yield up to of {15.3\%} gain in terms of the average of total VR QoS utility compared to the Q-learning for the case with 5 SBSs. In Fig. \ref{fig9b}, we can also see that the proposed ESN-based learning algorithm achieves, respectively, up to {17.1\%} and {36.7\%} improvements in terms of the total utility value compared to Q-learning and proportional fair algorithms for the case with 4 SBSs. These gains are due to the fact that our ESN algorithm can store the past ESN information and use it to build the relationship between the input and output. Hence, the proposed learning algorithm can predict the output (utility value) and, hence, find a better solution for allocating resources.

 \begin{figure}[!t]
  \begin{center}
   \vspace{0cm}
    \includegraphics[width=7cm]{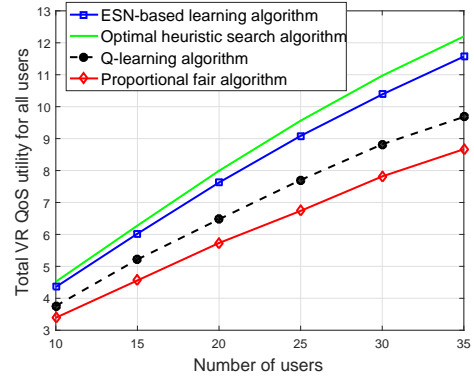}
    \vspace{-0.3cm}
    \caption{\label{figure10}The total VR QoS utility for all users vs. total number of users. Here, the total VR QoS utility refers to (\ref{eq:u}).}
  \end{center}\vspace{-0.7cm}
\end{figure}
\vspace{-0cm}

 \begin{figure}[!t]
  \begin{center}
   \vspace{0cm}
    \includegraphics[width=7cm]{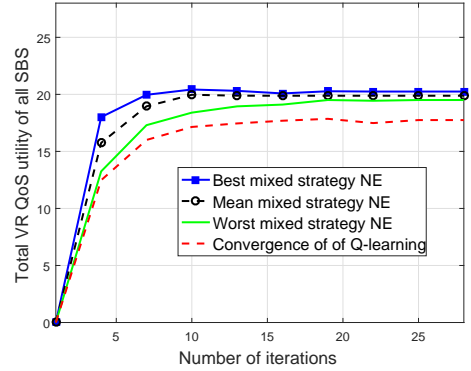}
    \vspace{-0.3cm}
    \caption{\label{figure6}{Convergence of the proposed algorithm and Q-learning. Here, total VR QoS utility refers to (\ref{eq:u})}}
  \end{center}\vspace{-0.7cm}
\end{figure}
\vspace{-0cm}

 \begin{figure}[!t]
  \begin{center}
   \vspace{0cm}
    \includegraphics[width=7cm]{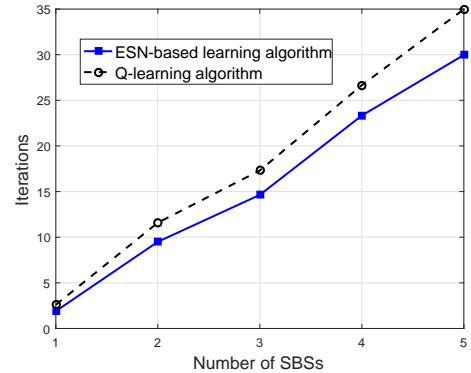}
    \vspace{-0.3cm}
    \caption{\label{figure7}The convergence time as a function of the number of SBSs. }
  \end{center}\vspace{-0.9cm}
\end{figure}
\vspace{-0cm}

\begin{figure*}[!t]
\centering
\vspace{0cm}
\subfigure[Optimal action of the ESN-based algorithm for the users over downlink and uplink]{
\label{figure8a} 
\includegraphics[width=6.5cm]{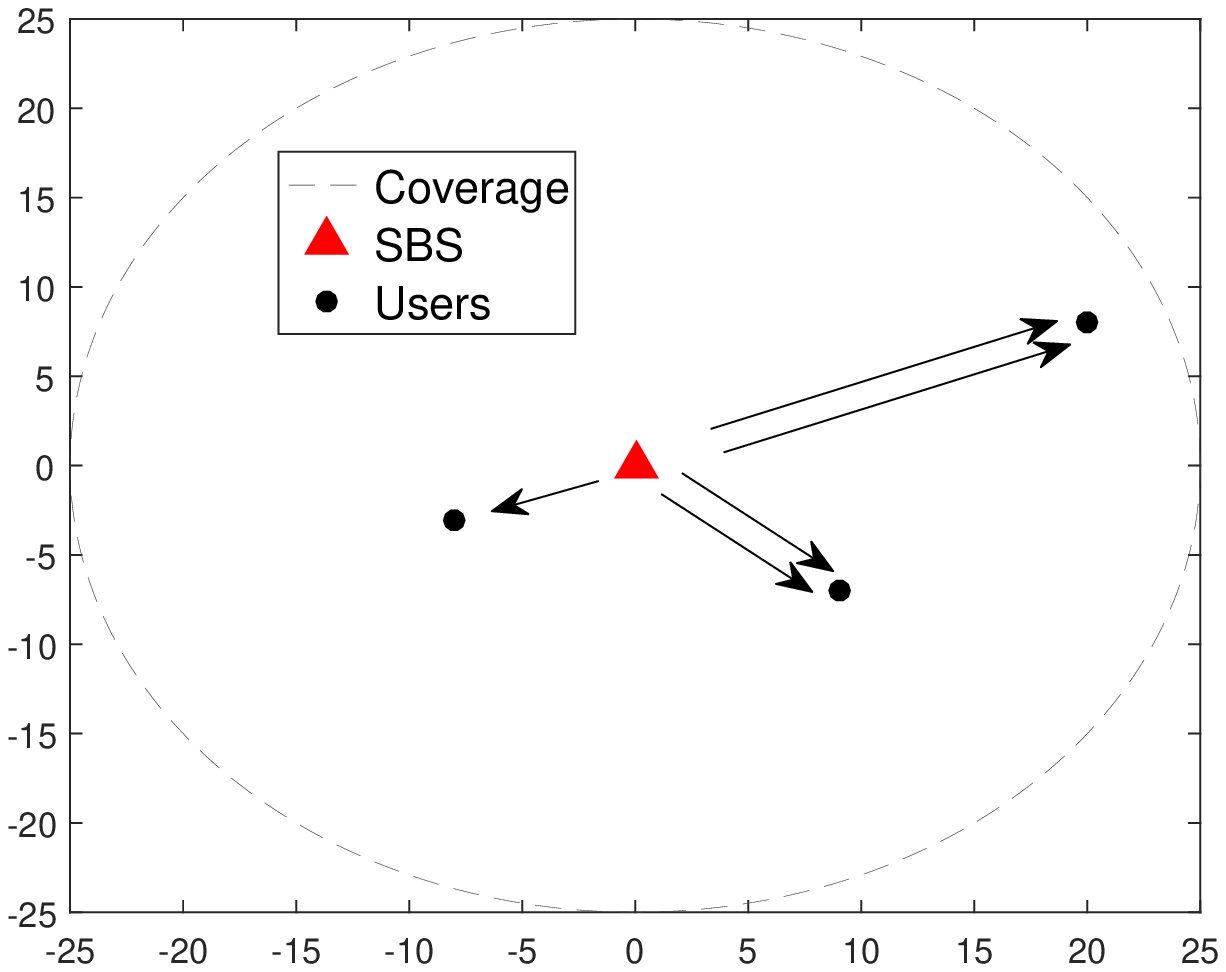} 
\label{figure8b}
\includegraphics[width=6.5cm]{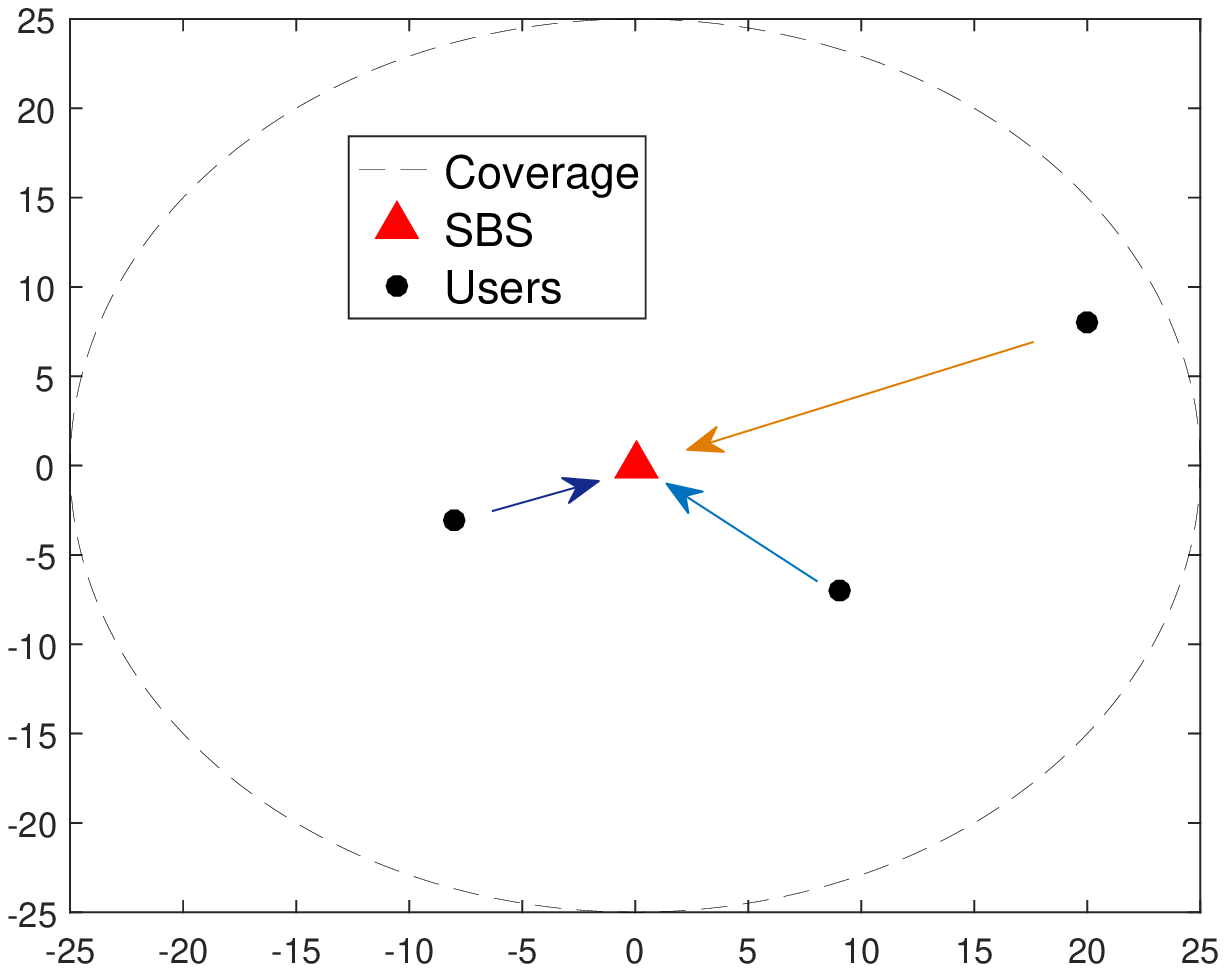}}
\subfigure[Optimal action of proportional fair algorithm for the users over downlink and uplink]{
\label{figure8c} 
\includegraphics[width=6.5cm]{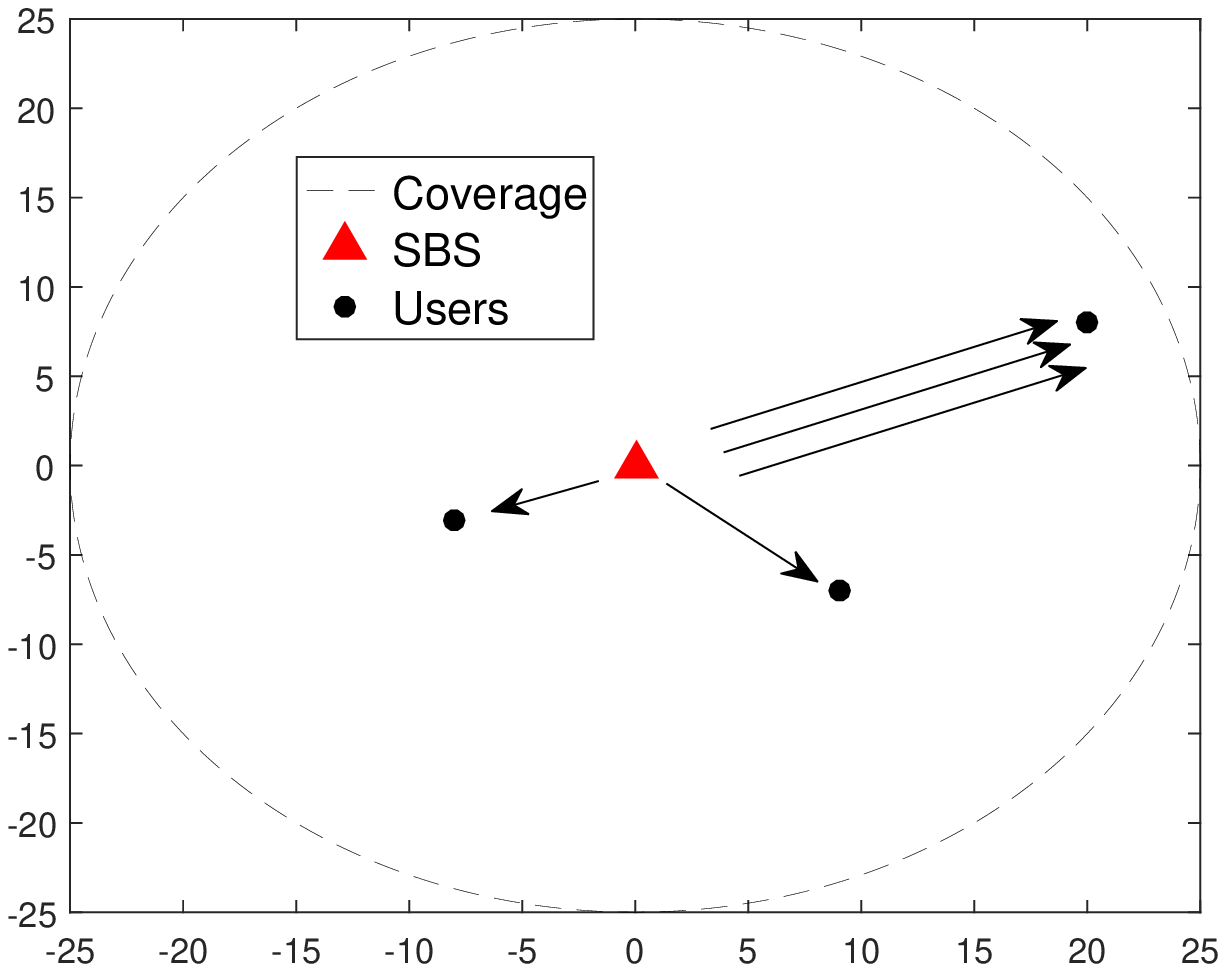} \includegraphics[width=6.5cm]{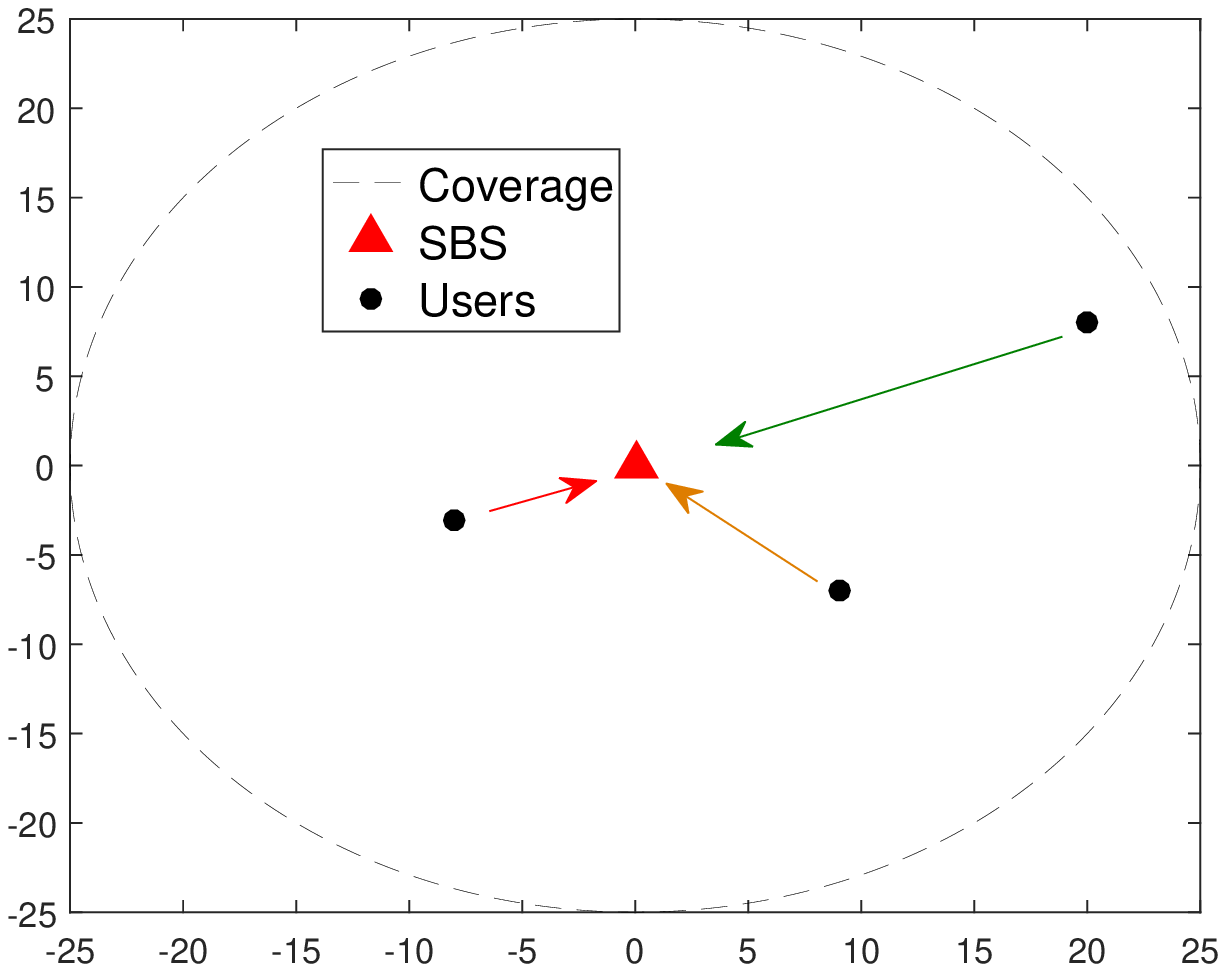}}
  \vspace{-0.1cm}
 \caption{\label{figure8} Optimal actions resulting from the different algorithms (Each black arrow represents the downlink resource blocks while each color arrow represents an group of uplink resource blocks).}
 \vspace{-0.6cm}
\end{figure*}

In Fig. \ref{figure10}, we show how the total utility value of VR QoS for all users changes as the total number of users varies. From Fig. \ref{figure10}, we can see that, as the number of users increases, the total utility values of all considered algorithms increase. This is due to the fact that, in all algorithms, each SBS has a limited coverage area and, hence, the number of users located in each SBS's coverage increases with the network size. Moreover, since each SBS has a limited number of resource blocks, the number of users that can associate with each SBS is also bounded. In particular, as the number of users located within the coverage of a given SBS exceeds the maximal number of users that each SBS can provide service, the SBS will only service the users that can maximize the total utility value. The VR QoS of the remaining users will be 0. In this case, the total utility value will also increase with the number of the users. From Fig. \ref{figure10}, we can also see that the proposed algorithm achieves, respectively, up to {22.2\%} and {37.5\%} gains in terms of the total utility value of VR QoS compared to Q-learning and proportional fair algorithms for the case with 35 users. 


Fig. \ref{figure6} shows the number of iterations needed till convergence for both Q-learning and the proposed ESN-based learning approach {with different mixed-strategy NEs} and Q-learning. In this figure, we can see that, as time elapses, the total VR QoS utilities for both the proposed algorithm and Q-learning increase until convergence to their final values. Fig. \ref{figure6} also shows that the proposed algorithm for best mixed strategy NE needs {19} iterations to reach convergence while Q-learning needs 25 iterations to reach convergence. Hence, the proposed algorithm achieves {24.2\%} gain in terms of the number of the iterations needed to reach convergence compared to Q-learning. This is because the ESN in the proposed algorithm can store the SBSs' action strategies and its corresponding total utility values. 
{Fig. \ref{figure6} also shows that, different mixed-strategy NEs will result in different total utility values of all SBSs. However, the total utility values achieved by these mixed-strategy NEs are very close to each other. Moreover, Fig. \ref{figure6} also shows that the worst mixed-strategy NE resulting from the proposed algorithm can still achieve 14\% gain of the total utility value compared to Q-learning. Hence, regardless of the reached NE, the total utility value of the proposed algorithm will be higher than the utility value achieved by Q-learning.}

In Fig. \ref{figure7}, we show how the convergence time changes as the number of SBSs varies. In this figure, we can see that as the number of the SBSs increases, the convergence time of both algorithms increases. Indeed, as the number of SBSs increases, the proposed ESN algorithm will require more time to accurately calculate the VR QoS utility. From Fig. \ref{figure7}, we can also see that as the number of SBSs increases, the difference in the convergence time between the proposed algorithm and Q-learning increases. This stems from the fact that as the number of SBSs increases, the number of actions for each SBS decreases and, hence, the number of output weight matrix used to predict the VR QoS utility for each action decreases.

Fig. \ref{figure8} shows the optimal actions resulting from the proposed ESN-based learning algorithm and proportional fair algorithm. Here, each color arrow represents a unique group of uplink resource blocks. From Figs. \ref{figure8a} and \ref{figure8c}, we can see that, for downlink resource allocation, the proportional fair algorithm allocates most of the downlink resource blocks to the user located farthest to the SBS while the proposed learning algorithm allocates only two groups of resource blocks to the farthest user. This is due to the fact that the proportional fair algorithm only considers the users' resource blocks demands while the proposed learning algorithm considers how to maximize the total utility values of VR QoS for all associated users. Figs. \ref{figure8a} and \ref{figure8c} also show that, both the proposed ESN-based learning algorithm and proportional fair algorithm allocate three groups of resource blocks to the farthest user. However, the uplink resource blocks allocated to each user are different. This is due to the fact that, in uplink, the proportional fair algorithm only considers the users' resource blocks demands while ignoring the uplink interference pertaining to the allocation of uplink resource blocks. In this case, the interference of users in uplink will significantly decrease the total VR QoS utility for each user. Note that, since each SBS allocates their all downlink and uplink resource blocks to its associated users, the interference of the users in downlink will not change as the actions vary while the interference in uplink depends on the actions.



\section{Conclusion}
In this paper, we have developed a novel multi-attribute utility theory based VR model that can capture the tracking and delay components of VR QoS. Based on this model, we have proposed a novel resource allocation framework for optimizing the VR QoS for all users. We have formulated the problem as a noncooperative game between the SBSs that seeks to maximize the average VR QoS utilities for all users. To solve this problem, we have developed a novel algorithm based on the machine learning tools of echo state networks. The proposed algorithm enables each SBS to decide on its actions autonomously according to the users' and networks' states. Moreover, the proposed learning algorithm only needs to update the mixed strategy during the training process and, hence, can quickly converge to a mixed-strategy NE. Simulation results have shown that the proposed VR model can capture the VR QoS in wireless networks while providing significant performance gains.

\section*{Appendix}
\subsection{Proof of Theorem \ref{th:1}} 
For \romannumeral1), the gain that stems from increasing the number of uplink resource blocks allocated to user $i$, $\Delta U_i^\textrm{u}$, is:
\begin{equation}\small\label{eq:Ud}
\begin{split}
\!\!\!\!\!\!&U_i\!\left(\! {{D_i}\!\!\left( \!{{\boldsymbol{s}_{ij}^\textrm{d},\boldsymbol{s}_{ij}^\textrm{u}\!\!+\!\!\Delta\boldsymbol{s}_{ij}^\textrm{u}}} \right)\!,{K_i}\!\!\left(\! \boldsymbol{s}_{ij}^\textrm{u}\!+\!\Delta\boldsymbol{s}_{ij}^\textrm{u} \right)}\! \right)\!\!-\!U_i\!\!\left( \!{{D_i}\!\!\left( \!{{\boldsymbol{s}_{ij}^\textrm{d},\boldsymbol{s}_{ij}^\textrm{u}}} \!\right)\!,{K_i}\!\left( \boldsymbol{s}_{ij}^\textrm{u} \right)} \!\right)\\
=&K_i\left(\boldsymbol{s}_{ij}^\textrm{u}+\Delta\boldsymbol{s}_{ij}^\textrm{u}\right)\times\frac{{{D_{\max ,i}\left( \boldsymbol{s}_{ij}^\textrm{u}+\Delta\boldsymbol{s}_{ij}^\textrm{u}\right)}   - {D_i}\left( {{\boldsymbol{s}_{ij}^\textrm{d},\boldsymbol{s}_{ij}^\textrm{u}+\Delta\boldsymbol{s}_{ij}^\textrm{u}}} \right)}}{{D_{\max ,i}\left( \boldsymbol{s}_{ij}^\textrm{u}+\Delta\boldsymbol{s}_{ij}^\textrm{u}\right) - \gamma_{D}}}
\\&-K_i\left(\boldsymbol{s}_{ij}^\textrm{u}\right)\times \frac{{{D_{\max ,i}\left( \boldsymbol{s}_{ij}^\textrm{u}\right)}   - {D_i}\left( {{\boldsymbol{s}_{ij}^\textrm{d},\boldsymbol{s}_{ij}^\textrm{u}}} \right)}}{{D_{\max ,i}\left( \boldsymbol{s}_{ij}^\textrm{u}\right) - \gamma_{D}}},
\end{split}
\end{equation}
where (\ref{eq:Uiu}) is obtained by substituting (\ref{eq:dt}) and (\ref{eq:totaldelay}) into (\ref{eq:Ud}). Since $K_i\left(x\right)$ is determined by the bit errors due to the wireless transmission, (\ref{eq:Uiu}) cannot be further simplified.

For \romannumeral2), the gain of changing the downlink resource blocks, $\Delta U_i^\textrm{d}$, can be given by: 
\begin{equation}\nonumber
\begin{split}
\Delta U_i^\textrm{d}&={K_i}\left( {\boldsymbol{s}_{ij}^\textrm{u}} \right) \times \frac{{{D_i}\left( {\boldsymbol{s}_{ij}^\textrm{d},\boldsymbol{s}_{ij}^\textrm{u}} \right) - {D_i}\left( {\boldsymbol{s}_{ij}^\textrm{d} + \Delta \boldsymbol{s}_{ij}^\textrm{d},\boldsymbol{s}_{ij}^\textrm{u}} \right)}}{{{D_{\max }}\left( {\boldsymbol{s}_{ij}^\textrm{u}} \right) - {\gamma _{D}}}}\\&={K_i}\left( {\boldsymbol{s}_{ij}^\textrm{u}} \right)\times \frac{{D_i^T\left( {\boldsymbol{s}_{ij}^\textrm{d}} \right) - D_i^T\left( {\boldsymbol{s}_{ij}^\textrm{d} + \Delta s_{ij}^\textrm{d}} \right)}}{{{D_{\max ,i}}\left( {s_{ij}^\textrm{u}} \right) - {\gamma _{D}}}},\\
&=\frac{{{K_i}\left( {\boldsymbol{s}_{ij}^\textrm{u}} \right)}}{{{D_{\max ,i}}\left( {\boldsymbol{s}_{ij}^\textrm{u}} \right) - {\gamma _{D}}}} \times \frac{{{L}{c_{ij}}\left( {\Delta \boldsymbol{s}_{ij}^\textrm{d}} \right)}}{{{c_{ij}}\left( {\boldsymbol{s}_{ij}^\textrm{d}} \right){c_{ij}}\left( {\boldsymbol{s}_{ij}^\textrm{d} + \Delta \boldsymbol{s}_{ij}^\textrm{d}} \right)}}\\
&=\frac{{{K_i}\left( {\boldsymbol{s}_{ij}^\textrm{u}} \right)L}}{{{D_{\max ,i}}\left( {\boldsymbol{s}_{ij}^\textrm{u}} \right)\! -\! {\gamma _{D}}}}\! \times\! \frac{{{c_{ij}}\left( {\Delta \boldsymbol{s}_{ij}^\textrm{d}} \right)}}{{{c_{ij}}{{\left( {\boldsymbol{s}_{ij}^\textrm{d}} \right)}^2}\! \!+\! {c_{ij}}\!\left( {\boldsymbol{s}_{ij}^\textrm{d}} \right){c_{ij}}\!\left( {\Delta \boldsymbol{s}_{ij}^\textrm{d}} \right)}}.
\end{split}
\end{equation}
Here, when ${c_{ij}}\!\left( {\Delta \boldsymbol{s}_{ij}^\textrm{d}} \right) \!\!\gg\!\!  {c_{ij}}\!\left( {\boldsymbol{s}_{ij}^\textrm{d}} \right)$, $\frac{{{c_{ij}}\left( {\Delta \boldsymbol{s}_{ij}^\textrm{d}} \right)}}{{{c_{ij}}{{\left( {\boldsymbol{s}_{ij}^\textrm{d}} \right)}^2} + {c_{ij}}\left( {\boldsymbol{s}_{ij}^\textrm{d}} \right){c_{ij}}\left( {\Delta \boldsymbol{s}_{ij}^\textrm{d}} \right)}} \!\!\approx \!\! \frac{1}{{{c_{ij}}\left( {\boldsymbol{s}_{ij}^\textrm{d}} \right)}}$, and, consequently, $\Delta U_i^\textrm{d}=\frac{{{K_i}\left( {\boldsymbol{s}_{ij}^\textrm{u}} \right){L}}}{{\left( {{D_{\max ,i}}\left( {\boldsymbol{s}_{ij}^\textrm{u}} \right) - {\gamma _{D}}} \right){c_{ij}}\left( {\boldsymbol{s}_{ij}^\textrm{d}} \right)}}$. 
Moreover, as ${c_{ij}}\!\left( \!{\Delta \boldsymbol{s}_{ij}^\textrm{d}} \!\right) \ll  {c_{ij}}\!\left(\! {\boldsymbol{s}_{ij}^\textrm{d}} \right)$, $\frac{{{c_{ij}}\left( {\Delta \boldsymbol{s}_{ij}^\textrm{d}} \right)}}{{{c_{ij}}{{\left( \!{\boldsymbol{s}_{ij}^\textrm{d}} \right)}^2}\! \!+ {c_{ij}}\!\left(\! {\boldsymbol{s}_{ij}^\textrm{d}} \right){c_{ij}}\left( {\Delta \boldsymbol{s}_{ij}^\textrm{d}} \right)}} \approx \frac{{{c_{ij}}\left( {\Delta \boldsymbol{s}_{ij}^\textrm{d}} \right)}}{{{c_{ij}}{{\left( {\boldsymbol{s}_{ij}^\textrm{d}} \right)}^2}}}$. Thus, $\Delta U_i^{\textrm{d}}\!=\!\frac{{{K_i}\left( {\boldsymbol{s}_{ij}^\textrm{u}} \right){L}{c_{ij}}\left( {\Delta \boldsymbol{s}_{ij}^\textrm{d}} \right)}}{{\left(\! {{D_{\max ,i}}\left(\! {\boldsymbol{s}_{ij}^\textrm{u}} \!\right) - {\gamma _{D}}}\! \right){c_{ij}}\!\left( \!{\boldsymbol{s}_{ij}^\textrm{d}} \!\right)}}$.
 This completes the proof.  

\subsection{Proof of Proposition \ref{pro:1}} \label{ap:a} 
To prove Proposition \ref{pro:1}, we first need to prove that the number of actions for the users over the downlink is $\footnotesize\left( \begin{array}{l}
{S^\textrm{d}} - 1\\
\left| {\mathcal{N}\left( {{V_j}} \right)} \right| - 1
\end{array} \right)$. Since the SBSs will allocate all downlink resource blocks to their associated users, the interference from each SBS to its associated users is unchanged when the actions change. For example, the interference when SBS $j$ allocates resource block 1 to user $i$ is the same as the interference when SBS $j$ allocates resource block 2 to user $i$. 
Therefore, we only need to consider the number of downlink resource blocks allocated to each user and, consequently, the number of actions for the users over the downlink is $\footnotesize \left( \begin{array}{l}
{S^\textrm{d}} - 1\\
\left| {\mathcal{N}\left( {{V_j}} \right)} \right| - 1
\end{array} \right)$. Then, we need to prove that the number of actions for the users over the uplink is $\footnotesize {\sum\limits_{\boldsymbol{n} \in \mathcal{N}\left( {{V_j}} \right)} {\prod\limits_{i = 1}^{V_{j} - 1} \left( \begin{array}{l}
{n_i}\\
{S^\textrm{u}} - \sum\nolimits_{k = 1}^{i - 1} {{n_i}} 
\end{array} \right) } }$. For each vector $\boldsymbol{n} $, SBS $j$ has $\footnotesize \left( \begin{array}{l}
{n_1}\\
{S^\textrm{u}}
\end{array} \right)$ actions to allocate the resource blocks to the first user. Based on the resource blocks allocated to the first user, SBS $j$ will have $\footnotesize \left( \begin{array}{l}
{n_2}\\
{S^\textrm{u}-n_1}
\end{array} \right)$ actions to allocate the resource blocks to the second user. For other associated users, the number of actions can be derived using a similar method as the method of the second user. Therefore, the number of actions for the users over uplink is $\footnotesize {\sum\limits_{\boldsymbol{n} \in \mathcal{N}\left( {{V_j}} \right)} {\prod\limits_{i = 1}^{V_{j} - 1} \left( \begin{array}{l}
{n_i}\\
{S^\textrm{u}} - \sum\nolimits_{k = 1}^{i - 1} {{n_i}} 
\end{array} \right) } }$, and, hence, the number of actions for the users over uplink and downlink is $\footnotesize {\left( \begin{array}{l}
{S^\textrm{d}} - 1\\
\left| {\mathcal{N}\left( {{V_j}} \right)} \right| - 1
\end{array} \right) \footnotesize {\sum\limits_{\boldsymbol{n} \in \mathcal{N}\left( {{V_j}} \right)} {\prod\limits_{i = 1}^{V_{j} - 1} \left( \begin{array}{l}
{n_i}\\
{S^\textrm{u}} - \sum\nolimits_{k = 1}^{i - 1} {{n_i}} 
\end{array} \right) } } }$. 
This completes the proof.  
\subsection{Proof of Theorem \ref{th:2}} \label{ap:b}

In order to prove this theorem, we first need to prove that the ESN-based learning algorithm converges to a constant value. Here, we do not know the exact value to which the proposed algorithm converges. Hence, our purpose is to prove that the proposed algorithm cannot diverge. 
Then, we derive the exact value to which the ESN converges. For \romannumeral1), based on \cite[Theorem 8]{26}, the conditions of convergence for an ESN are: a) The ESN is $k$-step unambiguous and b) The ESN-based learning process is $k$ order Markov decision process (MDP). Here, the definition of $k$-step unambiguous can be given by:

\begin{definition}
\emph{
Given an ESN with initial state $\boldsymbol{\mu}_{0,j}$, we assume that the input sequence $\boldsymbol{x}_{0,j},\ldots ,\boldsymbol{x}_{\tau,j}$ results in an internal state $\boldsymbol{\mu}_{\tau,j}$, and the input sequence $\boldsymbol{x}'_{0,j},\ldots, \boldsymbol{x}'_{\tau',j}$  results in an internal state $\boldsymbol{\mu}'_{\tau',j}$. If $\boldsymbol{\mu}_{\tau,j} = \boldsymbol{\mu}'_{\tau',j}$ implies that $\boldsymbol{x}_{\tau-i,j} = \boldsymbol{x}'_{\tau'-i,j}$, for all $i = 0, \ldots, \tau$, then the ESN is $k$-step unambiguous.} 
\end{definition}

Here, $\boldsymbol{u}_{\tau,j} = \boldsymbol{u}'_{\tau',j}$ can be rewritten as: 
 \begin{equation}\nonumber
 \begin{split}
 \boldsymbol{u}_{\tau,j}-\boldsymbol{u}'_{\tau',j}&={\boldsymbol{W}_j\left({\boldsymbol{\mu}_{\tau - 1,j}}-{\boldsymbol{\mu}'_{\tau' - 1,j}}\right) + \boldsymbol{W}_j^\textrm{in}\left({\boldsymbol{x}_{\tau,j}}-{\boldsymbol{x}'_{\tau',j}}\right)}\\
 &=\left[ {\begin{array}{*{20}{c}}
{{w_{11}}\left( {{\mu _{\tau  - 1,j1}} - {{\mu '}_{\tau ' - 1,j1}}} \right)}\\
 \vdots \\
{{w_{{N_w}{N_w}}}\left( {{\mu _{\tau  - 1,j{N_w}}} - {{\mu '}_{\tau ' - 1,j{N_w}}}} \right)}
\end{array}} \right]\\
&\;\;\;\;-\left[ {\begin{array}{*{20}{c}}
{\boldsymbol{W}_{j1}^\textrm{in}\left( {{\boldsymbol{x}_{\tau ,j}} - {\boldsymbol{x}'_{\tau' ,j}}} \right)}\\
 \vdots \\
{\boldsymbol{W}_{j{N_w}}^\textrm{in}\left( {{\boldsymbol{x}_{\tau ,j}} - {{\boldsymbol{x}'}_{\tau' ,j}}} \right)}
\end{array}} \right],
\end{split}
 \end{equation}
 where $\mu _{\tau  - 1,jk}$ is element $k$ of $\boldsymbol{\mu}_{\tau- 1,j} $ and $\mu' _{\tau'  - 1,jk}$ element $k$ of $\boldsymbol{\mu}'_{\tau' - 1,j}$. Since the tanh function in (\ref{eq:state}) ranges from -1 to 1, the maximum value of $\left( {{\mu_{\tau-1,jk}} - {{\mu}'_{\tau'-1,jk}}} \right)$ is 2. As $w_{kk} \in \left(-1,1\right), k=1,\ldots,N_w $, $\mathop {\max }\limits_k w_{kk}\left( {\mu_{\tau-1,jk} - \mu'_{\tau-1,jk}} \right)<2$. In this case, if ${\boldsymbol{W}_{j{k}}^\textrm{in}\left( {{\boldsymbol{x}_{\tau ,j}} - {{\boldsymbol{x}'}_{\tau ,j}}} \right)} \ge 2$, then $ \boldsymbol{\mu}_{\tau,j}-\boldsymbol{\mu}'_{\tau',j} \ne 0$. 
 Therefore, if $\boldsymbol{\mu}_{\tau,j}-\boldsymbol{\mu}'_{\tau',j} = 0$, then ${\boldsymbol{\mu}_{\tau - 1,j}}={\boldsymbol{\mu}'_{\tau' - 1,j}}$ and ${\boldsymbol{x}_{\tau,j}}={\boldsymbol{x}'_{\tau',j}}$. In this case, an ESN is $k$-step unambiguous when ${\boldsymbol{W}_{j{i}}^\textrm{in}\left( {{\boldsymbol{x}_{\tau ,j}} - {{\boldsymbol{x}'}_{\tau ,j}}} \right)} \ge 2$. Since the dynamic reservoir can only store limited ESN information \cite{chen2017machine}, the dynamic reservoir state $\boldsymbol{\mu}_{\tau,j}$ only depends on the finite past reservoir states, i.e., $\boldsymbol{\mu}_{\tau-1,j},\ldots, \boldsymbol{\mu}_{\tau-k,j}$. Moreover, the number of reservoir states and actions in the proposed algorithm is finite. Therefore, the proposed ESN-based algorithm is a $k$ order MDP and, hence, condition 2) is satisfied. 
 For case 2), if the learning rate of the proposed algorithm satisfies Robbins-Monro conditions and the proposed algorithm is a $k$ order MDP, the proposed algorithm will satisfy the conditions in \cite[Theorem 1]{Chen2016Echo} and, hence, converges to a region. 
 For both cases \romannumeral1) and \romannumeral2), based on \cite[Theorem 1]{Chen2016Echo}, the proposed ESN-based learning algorithm will converge to the utility value, $\boldsymbol{\hat u}_j$. This completes the proof.

\bibliographystyle{IEEEbib}
\bibliography{references1}
\end{document}